\documentclass[sigplan,9pt,authorversion]{acmart}
\usepackage[utf8]{inputenc}
\usepackage{mathpartir}
\usepackage{stmaryrd}
\usepackage{tikz}
\usepackage{blindtext}
\usepackage{listings}
\usepackage{mathrsfs}
\usepackage{mathtools}
\usepackage{tcolorbox}
\usepackage{calc}
\usepackage{enumitem}
\usepackage{tabu}
\usepackage{graphicx}
\usepackage{xifthen}
\usepackage[yyyymmdd]{datetime}
\usepackage[realmainfile]{currfile}
\usepackage{wrapfig}
\usepackage{wasysym}

\usetikzlibrary{arrows,chains,matrix,positioning,scopes,backgrounds,calc,
  cd,decorations.pathmorphing,fit}

\tikzstyle{guide}=[]
\tikzstyle{operator}=[draw, thick,minimum width=1cm,minimum height=.5cm]
\tikzstyle{wire}=[->]
\tikzstyle{junction}=[fill,circle,inner sep=1pt]
\tikzstyle{scaleboxinv}=[draw, thick, rounded corners, transparent]
\tikzstyle{scalebox}=[scaleboxinv, opaque]
\tikzstyle{focused}=[color=red,text=red]

\newcommand{\eqcounter}{\stepcounter{equation}$(\arabic{equation})$}
\newcolumntype{C}{>{\makebox[1em][r]{\eqcounter}}l}

\newcommand{\inferrulefactor}{.7}

\setlist[description]{font=\normalfont\itshape\textbullet\space}

\newtoggle{fullversion}

\ifcurrfile{\themainfile}{./generalized-modality-submit.tex}{
  \togglefalse{fullversion}
}{
  \toggletrue{fullversion}
}

\ifcurrfile{\themainfile}{generalized-modality-submit.tex}{
  \togglefalse{fullversion}
}{
}

\newtheorem{defn}{Definition}

\newtheorem{prop}{Property}
\newtheorem{coro}{Corollary}

\theoremstyle{definition}
\newtheorem{example}{Example}[section]



\newcommand{\IRULE}[3][]{
  \scalebox
  {\inferrulefactor}
  {
    \ifthenelse{\isempty{#1}}
    {\inferrule{#2}{#3}}
    {\inferrule[#1]{#2}{#3}}
  }
}

\iftoggle{fullversion}{
  \newcommand{\Refapp}[1]{Appendix~\ref{app:#1}}
}{
  \newcommand{\Refapp}[1]{the appendix}
}

\newcommand{\Reffig}[1]{Figure~\ref{fig:#1}}

\newcommand{\Refsec}[1]{Section~\ref{sec:#1}}

\newcommand{\Refprop}[1]{Property~\ref{prop:#1}}

\newcommand{\Reflemmshort}[1]{\ref{lemm:#1}}
\newcommand{\Reflemm}[1]{Lemma~\Reflemmshort{#1}}
\newcommand{\Refthm}[1]{Theorem~\ref{thm:#1}}
\newcommand{\Refcoro}[1]{Corollary~\ref{coro:#1}}
\newcommand{\Refdef}[1]{Definition~\ref{defn:#1}}

\newcommand{\Refex}[1]{Example~\ref{ex:#1}}

\newcommand{\Refeqshort}[1]{(\ref{eq:#1})}
\newcommand{\Refeq}[1]{Equation~\Refeqshort{#1}}
\newcommand{\Refruleshort}[1]{\rn{#1}}
\newcommand{\Refrule}[1]{rule~\Refruleshort{#1}}

\newcommand{\mt}[1]{\text{#1}}

\newcommand{\mi}[1]{\mathit{#1}}
\newcommand{\mb}[1]{\mathbf{#1}}

\newcommand{\CoreName}{{Core~$\lambda^{\ast}$}}
\newcommand{\Core}{\textsf{\CoreName}}

\newcommand{\defeq}{\triangleq}

\newcommand{\NP}{\omega}
\newcommand{\NI}{\omega+1}

\newcommand{\id}{\mathit{id}}

\makeatletter
\newcommand*{\bdiv}{%
  \nonscript\mskip-\medmuskip\mkern5mu%
  \mathbin{\operator@font div}\penalty900\mkern5mu%
  \nonscript\mskip-\medmuskip
}
\makeatother

\newcommand{\dom}[1]{\mathrm{dom}({#1})}

\newcommand{\GROUND}{\mathit{G}}
\newcommand{\SCAL}{\mathit{S}}

\newcommand{\Er}[1]{\ensuremath{\mathbf{U}(#1)}}

\newcommand{\Refines}[2]{\ensuremath{{#1} \sqsubset {#2}}}

\newcommand{\catname}[1]{\mathscr{#1}}

\newcommand{\catC}{\catname{C}}
\newcommand{\catD}{\catname{D}}

\newcommand{\Yo}{\mathbf{y}}

\newcommand{\Psh}[1]{\widehat{#1}}
\newcommand{\TT}{\ensuremath{\mathcal{S}}}

\newcommand{\catop}[1]{#1^{\mi{op}}}
\newcommand{\catpo}{\mathbf{P}}
\makeatletter
\def\slashedarrowfill@#1#2#3#4#5{%
  $\m@th\thickmuskip0mu\medmuskip\thickmuskip\thinmuskip\thickmuskip
\relax#5#1\mkern-7mu%
\cleaders\hbox{$#5\mkern-2mu#2\mkern-2mu$}\hfill
\mathclap{#3}\mathclap{#2}%
\cleaders\hbox{$#5\mkern-2mu#2\mkern-2mu$}\hfill
\mkern-7mu#4$%
}
\def\rightslashedarrowfill@{%
  \slashedarrowfill@\relbar\relbar\mapstochar\rightarrow}
\newcommand\xslashedrightarrow[2][]{%
  \ext@arrow 0055{\rightslashedarrowfill@}{#1}{#2}}
\makeatother
\newcommand{\catrel}{\xslashedrightarrow{}}

\newcommand{\lin}[1]{\overline{p}}

\newcommand{\Set}{\mathbf{Set}}
\newcommand{\Bool}{\mathbf{Bool}}
\newcommand{\One}{\mathbf{1}}

\newcommand{\iso}{\cong}

\newcommand{\oCLO}[2]{\mu_{X,Y}}

\newcommand{\WARP}{\mathcal{W}}
\newcommand{\EFF}{\mathcal{E}}
\newcommand{\PER}{\mathcal{P}}

\newcommand{\lLENGTH}[1]{\left\lvert#1\right\rvert}
\newcommand{\lWEIGHT}[1]{\left\lVert#1\right\rVert}

\newcommand{\wID}{\underline{\mi{id}}}
\newcommand{\wZERO}{\underline{0}}
\newcommand{\wONE}{\underline{1}}
\newcommand{\wPRED}{\underline{-1}}
\newcommand{\wSUCC}{\underline{+1}}
\newcommand{\wOM}{\underline{\omega}}
\newcommand{\wFIVE}{\underline{5}}

\newcommand{\opON}{\mathop{\ast}}

\newcommand{\colorKW}{black}
\newcommand{\colorTY}{black}
\newcommand{\colorIDT}{black}

\newcommand{\kw}[1]{\ensuremath{\mb{\color{\colorKW} #1}}}
\newcommand{\kwop}[1]{\mathop{\color{\colorKW}#1}}
\newcommand{\tyid}[1]{\mb{\color{\colorTY} #1}}
\newcommand{\tyop}[1]{\mathrel{\color{\colorTY} #1}}
\newcommand{\idt}[1]{\ensuremath{\mathsf{\color{\colorIDT}#1}}}

\newcommand{\pw}[1]{(#1)^{\mbox{\scalebox{0.6}{$\infty$}}}}
\newcommand{\upw}[2]{{#1}\,\pw{#2}}

\newcommand{\symW}{\mathop{\scalebox{1.5}{\raisebox{-0.2ex}{$\ast$}}}}
\newcommand{\symLATER}{\tyop{\raisebox{0pt}{$\blacktriangleright$}}}
\newcommand{\symSOONER}{\tyop{\raisebox{0pt}{$\blacktriangleleft$}}}
\newcommand{\symCONST}{\tyop{\raisebox{0pt}{$\blacksquare$}}}
\newcommand{\symARR}{\tyop{\to}}
\newcommand{\symPROD}{\tyop{\times}}
\newcommand{\symSUM}{\tyop{+}}

\newcommand{\symCONS}{\kwop{::}}

\newcommand{\W}[1]{{\symW}_{#1}}

\newcommand{\tyI}{\ensuremath{\tyid{Int}}}
\newcommand{\tyB}{\ensuremath{\tyid{Bool}}}
\newcommand{\symS}{\tyid{Stream}}
\newcommand{\tyS}[1]{\ensuremath{\symS\;{#1}}}

\newcommand{\tyW}[2]{\ensuremath{\mathop{\W{#1}}{#2}}}
\newcommand{\tyLATER}[1]{\ensuremath{\mathop{\symLATER}{#1}}}
\newcommand{\tySOONER}[1]{\ensuremath{\mathop{\symSOONER}{#1}}}
\newcommand{\tyCONST}[1]{\ensuremath{\mathop{\symCONST}{#1}}}
\newcommand{\tyARR}[2]{\ensuremath{{#1} \symARR {#2}}}
\newcommand{\tyPROD}[2]{\ensuremath{{#1} \symPROD {#2}}}
\newcommand{\tySUM}[2]{\ensuremath{{#1} \symSUM {#2}}}

\newcommand{\ctxempty}{\cdot}

\newcommand{\synFUN}[3]{\kw{fun}\,({#1} : {#2}).{#3}}
\newcommand{\synPROJ}[2]{\kw{proj}_{#1}\,{#2}}
\newcommand{\synINJ}[3]{\kw{inj}_{#1}^{#2}\,{#3}}
\newcommand{\synCASE}[5]
{\kw{case} \, {#1} \, \kw{of}
  \{ \synINJ{1}{}{#2}.{#3} \mid \synINJ{2}{}{#4}.{#5} \}}
\newcommand{\synREC}[3]{\kw{rec}~({#1} : {#2}).{#3}}

\newcommand{\synBY}[2]{\ensuremath{{#1}\;\kw{by}\;{#2}}}

\newcommand{\synLET}[4]
{\ensuremath{\kw{let}\;{#1} : {#2} = {#3}\;\kw{in}\;{#4}}}
\newcommand{\synLLET}[3]{\ensuremath{\kw{let}\;{#1} : {#2} = {#3}\;\kw{in}}}

\newcommand{\synCOEL}[2]{\ensuremath{{#1};{#2}}}
\newcommand{\synCOER}[2]{\ensuremath{{#1};{#2}}}

\newcommand{\synCONS}[2]{\ensuremath{\vC{#1}{#2}}}
\newcommand{\synHEAD}[1]{\ensuremath{\kw{head}~{#1}}}
\newcommand{\synTAIL}[1]{\ensuremath{\kw{tail}~{#1}}}

\newcommand{\coeID}{\ensuremath{\kw{id}}}
\newcommand{\coeS}[1]{\tyS{#1}}
\newcommand{\coeARR}[2]{\tyARR{#1}{#2}}
\newcommand{\coePROD}[2]{\tyPROD{#1}{#2}}
\newcommand{\coeSUM}[2]{\tySUM{#1}{#2}}
\newcommand{\coeW}[2]{\ensuremath{\mathop{\W{#1}}{#2}}}

\newcommand{\coeWRAP}{\kw{wrap}}
\newcommand{\coeUNWRAP}{\kw{unwrap}}
\newcommand{\coeCONCAT}[2]{\ensuremath{\kw{concat}^{{#1},{#2}}}}
\newcommand{\coeDECAT}[2]{\ensuremath{\kw{decat}^{{#1},{#2}}}}
\newcommand{\coeDISTB}[1]{\kw{dist}_{#1}}
\newcommand{\coeFACTB}[1]{\kw{fact}_{#1}}
\newcommand{\coeDISTP}{\coeDISTB{\symPROD}}
\newcommand{\coeFACTP}{\coeFACTB{\symPROD}}

\newcommand{\coeINFL}{\kw{inflate}}
\newcommand{\coeDEL}[2]{\ensuremath{\kw{delay}^{{#1},{#2}}}}

\newcommand{\structMap}[2]{\ensuremath{\Sigma({#1};{#2})}}

\newcommand{\vC}[2]{\ensuremath{{#1} \mathop{::} {#2}}}
\newcommand{\vCLO}[3]{({#1}.{#2})\{#3\}}
\newcommand{\vINJ}[2]{\ensuremath{\kw{inj}_{#1} {#2}}}
\newcommand{\vW}[2]{\kw{w}({#1},{#2})}
\newcommand{\vSTOP}{\kw{stop}}
\newcommand{\vTHUNK}[2]{\kw{box}({#1})\{{#2}\}}

\newcommand{\vTRUNC}[2]{\ensuremath{\lfloor\,{#2}\,\rfloor_{#1}}}

\newcommand{\jugTR}[3]{\ensuremath{\vTRUNC{#2}{#1} \Downarrow {#3}}}
\newcommand{\jugEV}[4]{\ensuremath{{#1};{#2} \Downarrow_{#3} {#4}}}
\newcommand{\jugEVC}[4]{\ensuremath{{#1} [{#2}] \Downarrow_{#3} {#4}}}
\newcommand{\jugITER}[7]
{\ensuremath{{#1};{#2};{#3};{#4} \Uparrow_{#5}^{#6} {#7}}}

\newcommand{\purge}[1]{\mathrm{purge}(#1)}

\newcommand{\reaV}[2]{\ensuremath{\mathcal{V}({#1})_{#2}}}
\newcommand{\reaC}[2]{\ensuremath{\mathcal{C}({#1})_{#2}}}
\newcommand{\reaE}[2]{\ensuremath{\mathcal{E}({#1};{#2})}}

\newcommand{\IMPL}[1]{\ensuremath{\raisebox{.2pt}{$\triangleright_{\hspace{0pt}{#1}}$}}}
\newcommand{\IMPLM}[2]{\IMPL{\,{#1};{#2}}}

\newcommand{\EQIMPL}[4]{\ensuremath{{#1} \vdash {#2} \bowtie {#3} : {#4}}}
\newcommand{\CTXEQ}[4]{\ensuremath{{#1} \vdash {#2} \cong_{\mi{ctx}} {#3} : {#4}}}

\newcommand{\sem}[1]{\ensuremath{\llbracket {#1} \rrbracket}}

\newcommand{\semCURRY}[1]{\ensuremath{\lambda\Big[{#1}\Big]}}

\newcommand{\semFIX}[1]{\mi{fix}_{#1}}

\newcommand{\semEQ}{\ensuremath{\mathrel{\cong}}}

\newcommand{\jugE}[4]{\ensuremath{{#1} \vdash {#2} \semEQ {#3} : {#4}}}
\newcommand{\jugEC}[4]{\ensuremath{{#1} \semEQ {#2} : {#3} \subty {#4}}}

\newcommand{\opsem}[1]{\ensuremath{\llparenthesis {#1} \rrparenthesis}}

\newcommand{\rn}[1]{\textsc{#1}}

\newcommand{\subty}{\mathrel{{<\hspace{-.7ex}\raisebox{.1ex}{\normalfont{:}}}}}

\newcommand{\jugT}[3]{{#1} \vdash {#2} : {#3}}
\newcommand{\jugC}[3]{\ensuremath{{#1} : {#2} \subty {#3}}}
\newcommand{\jugCE}[4]{\ensuremath{({#1},{#2}) : {#3} \equiv {#4}}}
\newcommand{\jugV}[3]{\ensuremath{{#1} : {#2} \mathrel{@} {#3}}}

\newcommand{\wD}{\mathop{\backslash}}
\newcommand{\wDn}{\mathop{\backslash_n}}

\newcommand{\NORM}[1]{\Lbag {#1} \Rbag}

\newcommand{\NORMin}[1]{\Lbag {#1} \Rbag_{\mt{in}}}
\newcommand{\NORMout}[1]{\Lbag {#1} \Rbag_{\mt{out}}}

\newcommand{\Prec}[2]{\ensuremath{\mi{Prec}({#1};{#2})}}

\newcommand{\SUP}[2]{{#1} \sqcup {#2}}
\newcommand{\SUPn}[2]{{#1} \sqcup_n {#2}}
\newcommand{\INF}[2]{{#1} \sqcap {#2}}
\newcommand{\INFn}[2]{{#1} \sqcap_n {#2}}

\newcommand{\GLC}{{$\mathbf{g}\lambda$-calculus}}

\newcommand{\algE}[3]{\ensuremath{{#1} \vdash {#2} \gg {#3}}}
\newcommand{\Restr}[2]{\ensuremath{{#1}|_{#2}}}
\newcommand{\minusW}[2]{\ensuremath{{#1} \ominus {#2}}}

\newcommand{\Elab}[2]{\ensuremath{\mi{Elab}({#1};{#2})}}
\newcommand{\Coe}[2]{\ensuremath{\mi{Coe}({#1};{#2})}}

\title{A Generalized Modality for Recursion}
\iftoggle{fullversion}{
  \subtitle{Extended Version}
}

\author{Adrien Guatto}
\affiliation{University of Bamberg}
\email{adrien@guatto.org}

\copyrightyear{2018}
\acmYear{2018}
\setcopyright{none}
\acmConference[LICS '18]{33rd Annual ACM/IEEE Symposium on Logic in Computer Science}{July 9--12, 2018}{Oxford, United Kingdom}
\acmBooktitle{LICS '18: 33rd Annual ACM/IEEE Symposium on Logic in Computer Science, July 9--12, 2018, Oxford, United Kingdom}
\acmPrice{15.00}
\acmDOI{10.1145/3209108.3209148}
\acmISBN{978-1-4503-5583-4/18/07}



\begin{document}
\sloppy

\begin{abstract}
  Nakano's~\emph{later} modality allows types to express that the output of a
function does not immediately depend on its input, and thus that computing its
fixpoint is safe.
This idea, guarded recursion, has proved useful in various contexts, from
functional programming with infinite data structures to formulations of
step-indexing internal to type theory.
Categorical models have revealed that the later modality corresponds in essence
to a simple reindexing of the discrete time scale.

Unfortunately, existing guarded type theories suffer from significant
limitations for programming purposes.
These limitations stem from the fact that the later modality is not expressive
enough to capture precise input-output dependencies of functions.
As a consequence, guarded type theories reject many productive definitions.

Combining insights from guarded type theories and synchronous programming
languages, we propose a new modality for guarded recursion.
This modality can apply any well-behaved reindexing of the time scale to a type.
We call such reindexings~\emph{time warps}.
Several modalities from the literature, including later, correspond to fixed
time warps, and thus arise as special cases of ours.

\end{abstract}

\maketitle

\section{Introduction}
\label{sec:introduction}

Consider the following piece of pseudocode.
\[
  \idt{nat} =
  \kw{fix}~\idt{natrec}
  ~\kw{where}~
  \idt{natrec}~\idt{xs} = 0 :: (\idt{map}~(\lambda \idt{x}.\idt{x}+1)~\idt{xs})
\]
This defines~\idt{nat}, the stream of natural numbers, as the fixpoint of a
function~\idt{natrec}.
How does one make sure that this definition is~\emph{productive}, in the sense
that the next element of~\idt{nat} can always be computed in finite time?

Guarded recursion, due to~\citet{Nakano-2000}, provides a type-based answer to
this question.
In type systems such as Nakano's, types capture precedence relationships between
pieces of data, expressed with respect to an implicit discrete time scale.
For example,~\idt{natrec} would receive the type
$
  \idt{natrec} : \tyARR{\tyLATER{\tyS{\tyI}}}{\tyS{\tyI}}.
$
The type~$\tyS{\tyI}$ describes streams which unfold in time at the rate of one
new element per step.
The~\emph{later}~($\symLATER$) modality shifts the type it is applied to one
step into the future; thus,~$\tyLATER{\tyS{\tyI}}$ also unfolds at the rate of
one element per step, but only starts unfolding after the first step.
Hence, the type of~\idt{natrec} expresses that the~\(n\)th element of its output
stream, which is produced at the~\(n\)th step, does not depend on the~\(n\)th
element of its input stream, since the latter arrives at the~\((n+1)\)th step.
This absence of~\emph{instantaneous} input-output dependence guarantees the
productivity of~$\kw{fix}~\idt{natrec}$.
Guarded recursion enforces this uniformly:~fixpoints are restricted to functions
with a type of the form~$\tyARR{\tyLATER{\tau}}{\tau}$.

Nakano's original insight has led to a flurry of
proposals~\cite{KrishnaswamiBenton-2011,
  BirkedalMogelbergSchwinghammerStovring-2012,
  AtkeyMcBride-2013,
  Krishnaswami-2013,
  Mogelberg-2014,
  BizjakBuggeCloustonMogelbergBirkedal-2016,
  CloustonBizjakBuggeBirkedal-2016,
  BirkedalBizjakCloustonBuggeGrathwohlSpittersVezzosi-2016,
  BahrBuggeGrathwohlMogelberg-2017,
  Severi-2017}.
Recent developments have integrated several advances---such as~\emph{clock
  variables}~\cite{AtkeyMcBride-2013} or the~\emph{constant}~($\symCONST$)
modality~\cite{CloustonBizjakBuggeBirkedal-2016}---into expressive languages
capturing many recursive definitions out of reach of more syntactic productivity
checks.
The~\emph{topos of trees}~\cite{BirkedalMogelbergSchwinghammerStovring-2012}
provides an elegant categorical setting where such languages find their natural
home.

Unfortunately, guarded recursion is currently limited by the inability of
existing languages to capture fine-grained dependencies.
Consider the following function, which returns a pair of streams.
\[
  \idt{natposrec} =
  \kw{fun}
  ~(\idt{xs},\idt{ys}).
  (0 :: \idt{ys},
  \idt{map}~(\lambda \idt{x}.\idt{x}+1)~\idt{xs})
\]
Its fixpoint is productive.
This can be seen in the table below, which gives the first iterations
of~$(\idt{nat},\idt{pos}) = \idt{natposrec}(\idt{nat},\idt{pos})$.
\[
  \begin{array}{l|cccccccccccc}
    \idt{nat}
    &
    \bot & 0::\bot & 0::\bot & 0::1::\bot & 0::1::\bot
    & \dots
    \\
    \idt{pos}
    &
    \bot & \bot & 1 :: \bot & 1::\bot & 1::2::\bot
    & \dots
  \end{array}
\]
Each stream grows infinitely often but only by one element every two steps.
The later modality, by itself, cannot capture this growth pattern, and thus this
program cannot be expressed as is in the guarded languages we know of.
Since~\idt{natposrec} is simply~\idt{natrec} modified to expose the result of a
subterm, this shows that existing systems can be overly rigid.
\citet[p.~12]{CloustonBizjakBuggeBirkedal-2016} give other examples hampered by
similar problems.

In our view, the example above does not indicate a problem with guarded
recursion~\emph{per se}, but rather illustrates the need for other temporal
modalities beyond later~(and constant).
Like later and constant, these new modalities would apply temporal
transformations onto types,~\emph{reindexing} them to change how much data is
available at each step.
In the example above, one would use a modality expressing growth at even time
steps, and another for growth at odd time steps.
Moreover, these new modalities should be interrelated, generalizing the known
interactions between later and constant.

\paragraph{Contribution}

In this paper, we propose a theory of temporal modalities subsuming later and
constant.
Rather than studying a fixed number of modalities, we merge all of them into a
single modality~$\symW$ parameterized by well-behaved reindexings of the
discrete time scale.
We call such reindexings~\emph{time warps} and speak of the~\emph{warping
  modality}.
The later and constant modalities correspond to specific time warps, and thus
arise as special cases of~$\symW$.

We build a simply-typed~$\lambda$-calculus,~\Core{}, around the warping
modality.
\Core{} integrates a notion of subtyping which internalizes the mathematical
structure of time warps.
We describe its operational semantics, as well as a denotational semantics in
the topos of trees.
We show that the type-checking problem for~\Core{} terms, while delicate because
of subtyping, is actually decidable.

\section{The Calculus}
\label{sec:calculus}



\subsection{Time Warps}

Let~$\NP$ denote the first infinite ordinal and~$\NI$ denote its successor,
which extends~$\NP$ with a maximal element~$\omega$.
It is technically convenient to see~$0$ as a special vacuous time step, and thus
we assume that~$\NP$ begins at~$1$.
(\citet{CloustonBizjakBuggeBirkedal-2016} follow the same convention.)
However,~$\NI$ still begins with~$0$.
If~$P$ is a preorder, we denote by~$\Psh{P}$ the set of its downward-closed
subsets ordered by inclusion, and write~$\Yo : P \to \Psh{P}$ for the order
embedding sending~$x$ to~$\{ x' \mid x' \le x \}$.
Observe that~$\Psh{\NP}$ is isomorphic to~$\NI$, with the isomorphism sending
the empty subset to~$0$ and the maximal subset~$\omega$ to itself in~$\NI$.
Thus, abusing notation, we write~$\Yo : \NP \to \NI$ for the map
sending positive natural numbers to their image in~$\NI$.

\begin{defn}[Time Warps]
  \label{defn:time-warps}
  A~\emph{time warp} is a cocontinuous~(sup-preserving) function from~$\NI$ to
  itself.
  Equivalently, it is a monotonic function~$p : \NI \to \NI$ such
  that~$p(0) = 0$ and~$p(\omega) = \bigsqcup_{n < \omega} p(n)$.
\end{defn}

We write~$p \le q$ when~$p$ is pointwise smaller than~$q$, that is,
when~$p(n) \le q(n)$ holds for all~$n$.
Given time warps~$p$ and~$q$, we write~$p \opON q$ for~$q \mathop{\circ} p$,
which is cocontinuous.
So is the identity function.
Moreover, function composition is left- and right-monotonic for the pointwise
order.
As a consequence,
\begin{prop}
  Time warps, ordered pointwise and equipped with composition, form a
  partially-ordered monoid, denoted~$\WARP$.
\end{prop}

The following time warps play a special role in our development.
\begin{mathpar}
  \wID(n) = n
  \and
  \wZERO(n) = 0
  \and
  \wPRED(n) = n-1
  \and
  \wOM(n) = \omega
\end{mathpar}
The definitions above are given for~$0 < n < \omega$ since the values at~$0$
and~$\omega$ follow from cocontinuity.
The time warps~$\wZERO$ and~$\wOM$ are respectively the least and greatest
elements of~$\WARP$.

\subsection{Syntax and Declarative Type System}

\Core{} is a two-level calculus distinguishing between~\emph{implicit terms}
and~\emph{explicit terms}.
Implicit terms correspond to source-level programs.
Explicit terms decorate implicit terms with type coercions.
Coercions act as proof terms for the subtyping
judgment~\citep{CurienGhelli-1990,BreazuTannenCoquandGunterScedrov-1991}.
They offer a convenient alternative to the manipulation of typing derivations in
non-syntax-directed type systems such as ours.

\paragraph{Ground types and scalars}

We assume given a finite set of ground types~$\GROUND$ and a family of pairwise
disjoint sets~$(S_\nu)_{\nu \in \GROUND}$.
The elements of~$S_\nu$ are the scalars~(ground values) of
type~$\nu \in \GROUND$.
We denote by~$s$ the elements
of~$\SCAL \defeq \bigcup_{\nu \in \GROUND} \SCAL_\nu$.

\paragraph{Types}

The types of~\Core{} are those of simply-typed~$\lambda$-calculus, including
products and sums, together with ground types, streams, and our warping
modality~$\W{p}$:
\[
  \tau
  \Coloneqq
  \nu
  \mid
  \tyS{\tau}
  \mid
  \tyARR{\tau}{\tau}
  \mid
  \tyPROD{\tau}{\tau}
  \mid
  \tySUM{\tau}{\tau}
  \mid
  \tyW{p}{\tau}.
\]
Informally,~$\tyW{p}{\tau}$ should be seen as a~``$p$-times'' faster version
of~$\tau$, in the sense of providing~$p$-times more data than~$\tau$ per step,
with the caveat that if~$p$ is less than~$\wID$,~$\tyW{p}{\tau}$ is
actually~``slower'' than~$\tau$.

\paragraph{Typing Contexts}

Typing contexts are lists of bindings~$x_i : \tau_i$ with the~$x_i$ pairwise
distinct.
We use~``$\ctxempty$'' for the empty context and~$\dom{\Gamma}$ for the finite
set of variables present in~$\Gamma$.
We write~$\Gamma(x)$ for the unique~$\tau$ such that~$(x : \tau)$ occurs
in~$\Gamma$, if it exists.

\paragraph{Explicit Terms}

\begin{figure}
  \begin{mathpar}
    \fbox{$\jugT{\Gamma}{e}{\tau}$}

    \IRULE[Var]{
      ~
    }{
      \jugT
      {\Gamma, x : \tau}
      {x}
      {\tau}
    }

    \IRULE[Fun]{
      \jugT{\Gamma, x : \tau_1}{e}{\tau_2}
    }{
      \jugT
      {\Gamma}
      {\synFUN{x}{\tau_1}{e}}
      {\tyARR{\tau_1}{\tau_2}}
    }

    \IRULE[App]{
      \jugT{\Gamma}{e_1}{\tyARR{\tau_1}{\tau_2}}
      \\
      \jugT{\Gamma}{e_2}{\tau_1}
    }{
      \jugT{\Gamma}{e_1~e_2}{\tau_2}
    }

    \IRULE[Pair]{
      \jugT{\Gamma}{e_1}{\tau_1}
      \\
      \jugT{\Gamma}{e_2}{\tau_2}
    }{
      \jugT{\Gamma}{(e_1, e_2)}{\tyPROD{\tau_1}{\tau_2}}
    }

    \IRULE[Proj$_{i \in \{ 1, 2 \}}$]{
      \jugT{\Gamma}{e}{\tyPROD{\tau_1}{\tau_2}}
    }{
      \jugT{\Gamma}{\synPROJ{i}{e}}{\tau_i}
    }

    \IRULE[Inj$_{i\in\{0, 1\}}$]{
      \jugT{\Gamma}{e}{\tau_{1+i}}
    }
    {
      \jugT
      {\Gamma}
      {\synINJ{1+i}{\tau_{2-i}}{e}}
      {\tySUM{\tau_1}{\tau_2}}
    }

    \IRULE[Case]{
      \jugT{\Gamma}{e}{\tySUM{\tau_1}{\tau_2}}
      \\
      \jugT{\Gamma, x_i : \tau_i}{e_i}{\tau} \mbox{ for } i \in \{ 1, 2 \}
    }{
      \jugT{\Gamma}
      {\synCASE{e}{x_1}{e_1}{x_2}{e_2}}
      {\tau}
    }

    \IRULE[Head]
    {
      \jugT
      {\Gamma}
      {e}
      {\tyS{\tau}}
    }
    {
      \jugT
      {\Gamma}
      {\synHEAD{e}}
      {\tau}
    }

    \IRULE[Tail]
    {
      \jugT
      {\Gamma}
      {e}
      {\tyS{\tau}}
    }
    {
      \jugT
      {\Gamma}
      {\synTAIL{e}}
      {\tyW{\wPRED}{\tyS{\tau}}}
    }

    \IRULE[Cons]
    {
      \jugT
      {\Gamma}
      {e_1}
      {\tau}
      \\
      \jugT
      {\Gamma}
      {e_2}
      {\tyW{\wPRED}{\tyS{\tau}}}
    }
    {
      \jugT
      {\Gamma}
      {\synCONS{e_1}{e_2}}
      {\tyS{\tau}}
    }

    \IRULE[Const]
    {
      s \in \SCAL_\nu
    }
    {
      \jugT
      {\Gamma}
      {s}
      {\nu}
    }

    \IRULE[Rec]{
      \jugT{\Gamma, x : \tyW{\wPRED}{\tau}}{e}{\tau}
    }{
      \jugT{\Gamma}{\synREC{x}{\tau}{e}}{\tau}
    }

    \IRULE[Warp]{
      \jugT{\Gamma}{e}{\tau}
    }{
      \jugT
      {\tyW{p}{\Gamma}}
      {\synBY{e}{p}}
      {\tyW{p}{\tau}}
    }

    \IRULE[SubR]
    {
      \jugT{\Gamma}{e}{\tau}
      \\
      \jugC{\alpha}{\tau}{\tau'}
    }
    {
      \jugT{\Gamma}{\synCOER{e}{\alpha}}{\tau'}
    }

    \IRULE[SubL]
    {
      \jugC{\beta}{\Gamma}{\Gamma'}
      \\
      \jugT{\Gamma'}{e}{\tau}
    }
    {
      \jugT{\Gamma}{\synCOEL{\beta}{e}}{\tau}
    }

    \IRULE[Struct]
    {
      \jugT{\Gamma}{e}{\tau}
      \\
      \sigma \in \structMap{\Gamma}{\Gamma'}
    }
    {
      \jugT{\Gamma'}{\sigma[e]}{\tau}
    }
  \end{mathpar}
  \caption{Typing Judgment}
  \label{fig:typing}
\end{figure}

The typing judgment for explicit terms~$e$ of~\Core{} is given
in~\Reffig{typing}.
Every typing rule from~\rn{Var} to~\rn{Case} is a standard one from
simply-typed~$\lambda$-calculus with products and sums.
We describe every other rule in turn, introducing the corresponding term formers
as we go.

The typing rules for stream destructors~(\synHEAD{e},~\synTAIL{e}) and the
stream constructor~(\synCONS{e_1}{e_2}) capture the fact that streams unfold at
the rate of one element per step.
As a consequence, the tail of a stream exists not now but~\emph{later}.
Since the later modality corresponds in our setting to the time warp~$\wPRED$,
the result of~\kw{tail} and the second argument of~$\kw{(::)}$ must be of
type~$\tyW{\wPRED}{\tyS{\tau}}$.

\Core{} terms include scalars from~$\SCAL$.
A scalar~$s$ is assigned the unique ground type~$\nu$ such
that~$s \in \SCAL_\nu$, as specified in rule~\rn{Const}.

Recursive definitions~$\synREC{x}{\tau}{e}$ follow the insight of Nakano: the
self-reference to~$x$ is only available later in the body~$e$, and thus here
receives type~$\tyW{\wPRED}{\tau}$~in rule~\rn{Rec}.

The term~$\synBY{e}{p}$ marks an introduction point for the warping modality.
Intuitively, it runs~$e$ in a local time scale whose relationship to the
surrounding time scale is goverened by the time warp~$p$:~the~$n$th tick of the
external time scale corresponds to the~$p(n)$th tick of the internal one.
Thus, assuming~$e$ has type~$\tau$,~$\synBY{e}{p}$ has type~$\tyW{p}{\tau}$.
This change in the amount of data produced comes at the price of a change in the
amount of data consumed:~the free variables of~$e$ should themselves be under
the~$\W{p}$ modality.
The context~$\tyW{p}{\Gamma}$ denotes~$\Gamma$ with~$\W{p}$ applied to each of
its types.

Explicit terms may include type coercions, applied either covariantly or
contravariantly.
Covariant coercion application~$\synCOER{e}{\alpha}$ applies the type
coercion~$\alpha$ to the result of~$e$.
Contravariant coercion application~$\synCOEL{\beta}{e}$ coerces the free
variables of~$e$ using the context coercion~$\beta$.
We will describe both kinds of coercions in a few paragraphs.


\paragraph{Structure maps}

Rule~\rn{Struct} is the only non-syntax-directed rule in our system.
It performs weakening, contraction, and exchange in a single step, depending on
the chosen~\emph{structure map} between
contexts~\cite{Atkey-2006,CurienFioreMunchMaccagnoni-2016}.
Structure maps~$\sigma \in \structMap{\Gamma}{\Gamma'}$ are functions
from~$\dom{\Gamma}$ to~$\dom{\Gamma'}$ such that~$\Gamma'(\sigma(x)) =
\Gamma(x)$ for all~$x \in \Gamma$.
The application of a structure map~$\sigma$, seen as a variable substitution, to
an explicit term~$e$ is written~$\sigma[e]$.

\paragraph{Type annotations}

Both~$\lambda$-abstractions and injections must contain type annotations.
This technical choice ensures that explicit terms are in Church style, and makes
their typing judgment essentially syntax directed~(up to rule~\rn{Struct}).

\paragraph{Type Coercions}

\begin{figure}
  \begin{mathpar}
    \fbox{$\jugC{\alpha}{\tau}{\tau'}$}

    \IRULE
    {
      ~
    }
    {
      \jugC{\coeID}{\tau}{\tau}
    }

    \IRULE
    {
      \jugC{\alpha_1}{\tau_1}{\tau_2}
      \\
      \jugC{\alpha_2}{\tau_2}{\tau_3}
    }
    {
      \jugC{\alpha_1;\alpha_2}{\tau_1}{\tau_3}
    }

    \IRULE
    {
      \jugC{\alpha}{\tau}{\tau'}
    }
    {
      \jugC
      {\coeS{\alpha}}
      {\tyS{\tau}}
      {\tyS{\tau'}}
    }

    \IRULE
    {
      \jugC{\alpha_1}{\tau_1'}{\tau_1}
      \\
      \jugC{\alpha_2}{\tau_2}{\tau_2'}
    }
    {
      \jugC
      {\coeARR{\alpha_1}{\alpha_2}}
      {\tyARR{\tau_1}{\tau_2}}
      {\tyARR{\tau_1'}{\tau_2'}}
    }

    \IRULE
    {
      \jugC{\alpha_1}{\tau_1}{\tau_1'}
      \\
      \jugC{\alpha_2}{\tau_2}{\tau_2'}
    }
    {
      \jugC
      {\coePROD{\alpha_1}{\alpha_2}}
      {\tyPROD{\tau_1}{\tau_2}}
      {\tyPROD{\tau_1'}{\tau_2'}}
    }

    \IRULE
    {
      \jugC{\alpha_1}{\tau_1}{\tau_1'}
      \\
      \jugC{\alpha_2}{\tau_2}{\tau_2'}
    }
    {
      \jugC
      {\coeSUM{\alpha_1}{\alpha_2}}
      {\tySUM{\tau_1}{\tau_2}}
      {\tySUM{\tau_1'}{\tau_2'}}
    }

    \IRULE
    {
      \jugC{\alpha}{\tau}{\tau'}
    }
    {
      \jugC
      {\coeW{p}{\alpha}}
      {\tyW{p}{\tau}}
      {\tyW{p}{\tau'}}
    }

    \IRULE
    {
      ~
    }
    {
      \jugCE
      {\coeWRAP}
      {\coeUNWRAP}
      {\tau}
      {\tyW{\wID}{\tau}}
    }

    \IRULE
    {
      ~
    }
    {
      \jugCE
      {\coeCONCAT{p}{q}}
      {\coeDECAT{p}{q}}
      {\tyW{p}{\tyW{q}{\tau}}}
      {\tyW{p \opON q}{\tau}}
    }

    \IRULE
    {
      ~
    }
    {
      \jugC
      {\coeINFL}
      {\nu}
      {\tyW{\wOM}{\nu}}
    }

    \IRULE
    {
      ~
    }
    {
      \jugCE
      {\coeDISTP}
      {\coeFACTP}
      {\tyW{p}{(\tyPROD{\tau_1}{\tau_2})}}
      {\tyPROD{\tyW{p}{\tau_1}}{\tyW{p}{\tau_2}}}
    }

    \IRULE
    {
      q \le p
    }
    {
      \jugC
      {\coeDEL{p}{q}}
      {\tyW{p}{\tau}}
      {\tyW{q}{\tau}}
    }
  \end{mathpar}
  \caption{Subtyping Judgment}
  \label{fig:subtyping}
\end{figure}

A coercion~$\jugC{\alpha}{\tau}{\tau'}$ performs a type conversion, transforming
input values of type~$\tau$ into output values of type~$\tau'$.
The rules for this syntax-directed subtyping judgment are given
in~\Reffig{subtyping}, where~$\jugCE{\alpha}{\alpha'}{\tau}{\tau'}$ is a
shorthand for~$\jugC{\alpha}{\tau}{\tau'}$ and~$\jugC{\alpha'}{\tau'}{\tau}$.
They fit into three groups.

The first group contains the identity coercion and sequential coercion
composition.
The identity coercion~$\coeID$ does nothing.
Two coercions~$\alpha_1$ and~$\alpha_2$ can be composed to
obtain~$\alpha_1;\alpha_2$, assuming the output type of~$\alpha_1$ matches the
input type of~$\alpha_2$.

The second group contains one coercion former for each type former.
Such coercions allow us to coerce values in depth.
Their typing rules express that subtyping is a congruence for all type formers
in the language.

The third group is where the interest of our subtyping relationship lies.
It contains coercions reflecting the mathematical structure of the warping
modality as subtyping axioms, including its interaction with other type formers.
This group can be divided again, now between several invertible coercions and a
non-invertible one.
\begin{itemize}
\item
  Coercions~$\coeWRAP$,~$\coeUNWRAP$,~$\coeCONCAT{p}{q}$, and~$\coeDECAT{p}{q}$
  reflect the monoidal structure of time warps at the type level.
  The coercions~$\coeDISTP$ and~$\coeFACTP$ ensure that the warping modality
  commutes with products.
  The coercion~$\coeINFL$ expresses that ground types stay constant through
  time, i.e.,~$\nu \subty \tyW{\wOM}{\nu}$.

\item
  The remaining coercion,~$\jugC{\coeDEL{p}{q}}{\tyW{p}{\tau}}{\tyW{q}{\tau}}$,
  reflects the ordering of time warps.
  Intuitively, it pushes data further into the future, and must thus ensure
  that~$p(n) \ge q(n)$ at any step~$n$.
  Its action cannot be undone when~$p \ne q$; for example, the
  coercion~$\jugC{\kw{next} \defeq
    \coeWRAP;\coeDEL{\wID}{\wPRED}}{\tau}{\tyW{\wPRED}{\tau}}$ has no inverse.
  (This coercion appears an an operator in some guarded type
  theories~\cite[\dots]{BirkedalMogelbergSchwinghammerStovring-2012,CloustonBizjakBuggeBirkedal-2016}.)
\end{itemize}
Note that we did not need to introduce an explicit inverse for~$\coeINFL$ since
one is already derivable
as~$\jugC{\coeDEL{\wOM}{\wID};\coeUNWRAP}{\tyW{\wOM}{\nu}}{\nu}$.

\paragraph{Context Coercions}

A context coercion~$\beta$ is a finite map from variables to coercions.
We have~$\jugC{\beta}{\Gamma}{\Gamma'}$
iff~$\dom{\Gamma} = \dom{\Gamma'} \subseteq \dom{\beta}$, and for every
variable~$x \in \dom{\Gamma}$,~$\beta(x)$ coerces~$\Gamma(x)$ into~$\Gamma'(x)$.
Context subtyping preserves the order of bindings.
This definition implies that rule~\rn{SubL} in~\Reffig{typing} can only be
applied when~$\beta(x)$ is defined for every free variable~$x$ of~$e$.

\paragraph{Implicit Terms and Erasure}

We define implicit terms, denoted~$t$, as explicit terms that do not contain any
coercions.
Each explicit term~$e$ thus corresponds to a unique implicit term obtained by
removing every coercion present in~$e$.
We adopt the notations of~\citet{MelliesZeilberger-2015}, and write~$\Er{e}$ for
this implicit term.
We also write that~$e$~\emph{refines}~$t$, noted~$\Refines{e}{t}$,
when~$\Er{e} = t$.

An implicit term is well-typed simply if it has a well-typed refiner, in the
sense expressed by the definition below.
\begin{equation}
  \mathrm{\fbox{$\jugT{\Gamma}{t}{\tau}$}}
  ~\Leftrightarrow~
  \exists e \sqsubset t, \jugT{\Gamma}{e}{\tau}
\end{equation}

\subsection{Type-Checking Explicit Terms}

The language of coercions and explicit terms enjoys uniqueness of typing.
The following result reflects this fact for coercions.

\begin{prop}[Uniqueness of Types for Coercions]
  \label{prop:uniqueness-of-types-coercions}
  For any coercion~$\alpha$, for any type~$\tau$~(resp.~$\tau'$) there is at
  most one type~$\tau'$~(resp.~$\tau$) such that~$\jugC{\alpha}{\tau}{\tau'}$
  holds.
\end{prop}

The analogous result for explicit terms is more delicate since rule~\rn{Struct}
is not syntax-directed.
Furthermore, the rule is not admissible: its use is sometimes required in order
to be able to use another rule.
One can prove that there are only three cases where~\Refrule{Struct} is
needed, establishing the following result.

\begin{prop}
  \label{prop:canonical-derivations}
  Any well-typed explicit term~$\jugT{\Gamma}{e}{\tau}$ has a canonical
  derivation where~\Refrule{Struct} is only used exactly once before every
  instance of rules~\rn{Var},~\rn{Warp}, and~\rn{SubL}.
\end{prop}

This characterization provides almost immediately an abstract deterministic
algorithm to type-check explicit terms; see~\Refapp{proofs} for details.
Its correctness implies the expected result.
\begin{prop}[Uniqueness of Types for Explicit Terms]
  \label{prop:uniqueness-of-types-explicit-terms}
  For any fixed~$\Gamma$ and~$e$, there is at most one type~$\tau$ such
  that~$\jugT{\Gamma}{e}{\tau}$ holds.
\end{prop}

Type-checking an~\emph{implicit} term~$t$ is a much harder problem, since it
involves finding a well-typed refiner~$\Refines{e}{t}$.
Moreover, this refiner should be canonical in a certain sense.
We study it in~\Refsec{algorithmic-type-checking}.

\subsection{Examples}
\label{sec:examples}


We finish this section with a few examples illustrating the type system, given
mostly as implicit terms.
We also assume that ground values and types include the integers.

\begin{example}[Constant Stream]
  \label{ex:constant-stream}
  This prototypical example defines a constant stream of zeroes.
  \begin{mathpar}
    \synREC
    {\idt{zeroes}}
    {\tyS{\tyI}}
    {(\synCONS{0}{\idt{zeroes}})}
    :
    \tyS{\tyI}
  \end{mathpar}
  This works as in other guarded-recursive languages:~the stream
  constructor~$(\symCONS)$ expects its second argument to have a type of the
  form~$\tyW{\wPRED}{\tyS{\tau}}$~($\tyLATER{\tyS{\tau}}$ in other guarded type
  theories), which is exactly the one provided by guarded recursion.
\end{example}

\begin{example}[Non-productive Stream]
  \label{ex:nonproductive-stream}
  The non-productive definition below does not define a stream, which by
  definition should hold an infinite number of values.
  \begin{mathpar}
    \synREC
    {\idt{nothing}}
    {\tyS{\tyI}}
    {\idt{nothing}}
    {\color{red}\mbox{ -- ill-typed!}}
  \end{mathpar}
  This definition is ill-typed in~\Core{} since, in the absence of a
  coercion~$\tyW{\wPRED}{\tyS{\tyI}} \subty \tyS{\tyI}$, we cannot
  apply~\Refrule{Rec}.
\end{example}

\begin{example}[Silent Stream]
  \label{ex:silent-stream}
  While~\Refex{nonproductive-stream} does not define a real stream holding an
  infinite number of values, it could be said to define a~``silent'' stream
  holding no value at all.
  Such streams are captured in~\Core{} as inhabitants
  of~$\tyW{\wZERO}{\tyS{\tau}}$, e.g.,
  \begin{mathpar}
    \synREC
    {\idt{nothing}}
    {\tyW{\wZERO}{\tyS{\tyI}}}
    {\idt{nothing}}
    :
    \tyW{\wZERO}{\tyS{\tyI}}.
  \end{mathpar}
  This definition is well-typed since the explicit
  term~$\synREC{\idt{nothing}}{\tyW{\wZERO}{\tyS{\tyI}}}{(\idt{nothing};\coeCONCAT{\wPRED}{\wZERO})}$
  refines it and has the expected type.
  Here,~$\coeCONCAT{\wPRED}{\wZERO}$ coerces values of
  type~$\tyW{\wPRED}{\tyW{\wZERO}{\tyS{\tyI}}}$ to values of
  type~$\tyW{\wPRED \opON \wZERO}{\tyS{\tyI}} =
  \tyW{\wZERO}{\tyS{\tyI}}$.
\end{example}

\Refex{silent-stream} illustrates how~\Core{} shifts the focus away from
productivity, seen as a yes-or-no question, to a more quantitative aspect of
program execution:~the amount of data produced.
Other warps make it possible to capture other forms of partial definitions,
beyond completely silent streams.
For example, writing~$\wFIVE$ for the warp sending any finite~$n$ to~$5$, the
type~$\tyW{\wFIVE}{\tyS{\tyI}}$ describes streams containing only~$5$ elements,
all of them available starting at the first time step.
The type system of~\Core{} ensures that the non-existent elements in such
partial streams can never be accessed; in particular, in a well-typed program
deconstructing a silent stream~\idt{xs}~(via~\kw{head} or~\kw{tail}) can only
happen under a context of the form~$C_1[\synBY{C_2[-]}{\wZERO}]$.
We will see in~\Refsec{operational-semantics} that the expression~$e$
in~$\synBY{e}{\wZERO}$ is never actually executed.

\begin{example}
  \label{ex:map}
  The example below implements the classic higher-order function~\idt{map} on
  streams, specialized for streams of integers since~\Core{} is monomorphic.
  As usual,~$\synLET{x}{\tau}{t_1}{t_2}$ is shorthand
  for~$(\synFUN{x}{\tau}{t_2})~t_1$.
  We assume that function application and~\kw{by} bind tighter than stream
  construction~($\symCONS$).
  \[
    \begin{array}{@{}l@{}}
      \kw{rec}~(\idt{map} :
      (\tyARR{\tyI}{\tyI}) \symARR \tyARR{\tyS{\tyI}}{\tyS{\tyI}}).
      \\
      \kw{fun}~(\idt{f} : \tyARR{\tyI}{\tyI})~(\idt{xs} : \tyS{\tyI}).
      \\
      \synLET
      {\idt{ys}}
      {\tyW{\wPRED}{\tyS{\tyI}}}
      {\synTAIL{\idt{xs}}}
      {
        \synCONS
        {
          \idt{f}\,(\synHEAD{\idt{xs}})
        }
        {
          \synBY
          {
            (\idt{map}\,\idt{f}\,\idt{ys})
          }
          {\wPRED}
        }
      }
    \end{array}
  \]
  Here, using~\kw{by} allows us to temporarily remove the~$\W{\wPRED}$ modality
  from the type of~\idt{map} in order to perform the recursive call.
  Rule~\rn{Warp} requires~\idt{map},~\idt{ys}, and~\idt{f} to have types of the
  form~$\tyW{\wPRED}{\tau}$.
  This is already the case for~\idt{map} and~\idt{ys}, and can be achieved
  for~\idt{f} using rule~\rn{SubL} with a context coercion sending~\idt{f} to
  $\jugC{\kw{next}}{\tyARR{\tyI}{\tyI}}{\tyW{\wPRED}{(\tyARR{\tyI}{\tyI})}}$
  and the other variables to~$\coeID$.
\end{example}

\begin{example}
  \label{ex:constant}
  The definition given in~\Refex{map}, since it is closed, can be put inside a
  local time scale driven by~$\wOM$.
  It thus receives the
  type~$\tyW{\wOM}{((\tyARR{\tyI}{\tyI}) \symARR
    \tyARR{\tyS{\tyI}}{\tyS{\tyI}})}$.
  Such a type is in effect not subject to the context restriction
  in~\Refrule{Warp}, since for any~$p$ we
  have~$\tyW{\wOM}{\tau} \subty \tyW{p}{\tyW{\wOM}{\tau}}$.
  Thus,~$\W{\wOM}$ corresponds to the~\emph{constant}~($\symCONST$) modality
  used in some guarded type theories~\cite{CloustonBizjakBuggeBirkedal-2016}.
\end{example}

In the remaining examples, we represent certain time warps as running sums of
ultimately periodic sequences of numbers, following ideas from
n-synchrony~\cite{CohenDurantonEisenbeisPagettiPlateau-POPL-2006,Plateau-2010}.
For example, the sequence~$\pw{1\,0}$ represents the time warp sending~$2n$
to~$n$ and~$2n+1$ to~$n+1$ for any finite positive~$n$, while the
sequence~$\pw{0\,1}$ represents the time warp sending both~$2n$ and~$2n+1$
to~$n$.
All the time warps we have used up to now can be represented in this
way:~$\wID$,~$\wZERO$,~$\wPRED$, and~$\wOM$ are represented
by~$\pw{1}$,~$\pw{0}$,~$\upw{0}{1}$, and~$\upw{\omega}{0}$ respectively.

\begin{example}[Mutual Recursion]
  As announced in~\Refsec{introduction}, the streams of natural and positive
  numbers can be defined in a guarded yet mutually-recursive way in~\Core{}.
  This is achieved by reflecting the rate at which each stream grows during a
  fixpoint computation within its type.
  (In the definition below, we represent the time warp~$\wPRED$ by the
  sequence~$\upw{0}{1}$ for consistency; in particular, the types
  of~$(\symCONS)$
  becomes~$\tau \symARR \tyW{\upw{0}{1}}{\tyS{\tau}} \symARR \tyS{\tau}$.)
  \[
  \begin{array}{@{}l}
    \kw{rec}~
    \idt{natpos} :
    \tyPROD
    {\tyW{\pw{1\,0}}{\tyS{\tyI}}}
    {\tyW{\pw{0\,1}}{\tyS{\tyI}}}.
    \\
    \synLLET
    {\idt{nat}}
    {\tyW{\upw{0}{1}}{\tyW{\pw{1\,0}}{\tyS{\tyI}}}}
    {\synPROJ{1}{\idt{natpos}}}
    \\
    \synLLET
    {\idt{pos}}
    {\tyW{\upw{0}{1}}{\tyW{\pw{0\,1}}{\tyS{\tyI}}}}
    {\synPROJ{2}{\idt{natpos}}}
    \\
    (\synBY
    {
      (
      \synCONS
      {0}
      {\idt{pos}}
      )
    }{\pw{1\,0}},
    \synBY
    {
      (
      \idt{map}~
      (\synFUN{\idt{x}}{\tyI}{\idt{x} + 1})~
      \idt{nat}
      )
    }{\pw{0\,1}}
  \end{array}
  \]
  The uses of projections are well-typed since the warping modality distributes
  over products via the~$\coeDISTP$ coercion.
  We assume that~\idt{map} has received the type given in~\Refex{constant}, and
  thus its use below~$\kw{by}~\pw{0\,1}$ is well-typed.
  Since~$\upw{0}{1} \opON~\pw{1\,0} = \pw{0\,1}$, coercing~\idt{nat}
  by~$\coeCONCAT{\upw{0}{1}}{\pw{1\,0}}$ gives the
  type~$\tyW{\pw{0\,1}}{\tyS{\tyI}}$, which lets us use~\idt{nat} with
  type~$\tyS{\tyI}$ under~$\kw{by}~\pw{0\,1}$.
  For~\idt{pos},
  since~$\upw{0}{1} \opON~\pw{0\,1} = \upw{0}{0\,1} = \pw{1\,0}
  \opON~\upw{0}{1}$, applying the
  coercion~$\coeCONCAT{\upw{0}{1}}{\pw{0\,1}};\coeDECAT{\pw{1\,0}}{\upw{0}{1}}$
  lets us use~\idt{pos} with type~$\tyW{\upw{0}{1}}{\tyS{\tyI}}$
  below~$\kw{by}~\pw{1\,0}$.
\end{example}

\begin{example}
  \label{ex:thue-morse-abridged}
  \citet[Example~1.10]{CloustonBizjakBuggeBirkedal-2016} present
  the~\emph{Thue-Morse sequence} as a recursive stream definition which is
  difficult to capture in guarded calculi.
  The productivity of this definition follows from the fact that a certain
  auxiliary stream function~\idt{h} produces two new elements of its output
  stream for each new element of its input stream.
  In~\Core{}, \idt{h} can be given
  type~$\tyARR{\tyS{\tyB}}{\tyW{\pw{2}}{\tyS{\tyB}}}$, allowing us to implement
  the Thue-Morse sequence with guarded recursion.
  (See~\Refapp{supplementary-material}.)
\end{example}

\section{Operational Semantics}
\label{sec:operational-semantics}

In this section, we present an operational semantics for explicit terms in
the form of a big-step, call-by-value evaluation judgment.
Intuitively, the evaluation judgment~$\jugEV{e}{\gamma}{n}{v}$ expresses that
the value~$v$ is a finite prefix of length~$n$ of the possibly infinite result
computed by~$e$ in the environment~$\gamma$.
We will say that the evaluation of~$e$ occurred~``at step~$n$'', or
simply~``at~$n$'' following the intuition that~$n$ is a Kripke world.
Another intuition is that this judgment describes an interpreter receiving a
certain amount~$n$ of~``fuel'' which controls how many times recursive
definitions have to be unrolled~\citep{AminRompf-2017}.

In most fuel-based definitional interpreters, the fuel parameter only decreases
along evaluation, typically by one unit at each recursive unfolding.
In our case, its evolution is much less constrained:~the amount of fuel may
actually increase or decrease by an arbitrary amount many times during the
execution of a single term.
This behavior follows from the presence of time warps:~to
evaluate~$\synBY{e}{p}$ at~$n$, one evaluates~$e$ at~$p(n)$.
Nevertheless, we show that the evaluation of a well-typed term always
terminates regardless of the quantity of fuel initially provided.

Since the evaluation of a term at~$n$ might involve the evaluation of one of its
subterms at~$p(n)$ with~$p$ an arbitrary warp,~we may need to evaluate a term
at~$0$ or~$\omega$.
The former case is dealt with using a dummy value~$\vSTOP$ which inhabits all
types at~$0$.
The latter case might seem problematic, as evaluating a term at~$\omega$ should
intuitively result in an infinite object rather than a finite one.
We represent such results by suspended computations~(\emph{thunks}) to be forced
only when used at a finite~$n$.
This is a standard operational interpretation of the constant
modality~\cite{BiermanDePaiva-2000,CloustonBizjakBuggeBirkedal-2016}.

\subsection{Values and Environments}

\begin{figure}
  \begin{mathpar}
    \fbox{\jugV{v}{\tau}{n}}

    \IRULE[VScal]
    {
      s \in \SCAL_\nu
    }
    {
      \jugV{s}{\nu}{n+1}
    }

    \IRULE[VCons]
    {
      \jugV{v_1}{\tau}{n+1}
      \\
      \jugV{v_2}{\tyS{\tau}}{n}
    }
    {
      \jugV
      {\vC{v_1}{v_2}}
      {\tyS{\tau}}
      {n+1}
    }

    \IRULE[VClosure]
    {
      \jugT{\Gamma, x : \tau_1}{e}{\tau_2}
      \\
      \jugV{\gamma}{\Gamma}{n}
    }
    {
      \jugV
      {\vCLO{x}{e}{\gamma}}
      {\tyARR{\tau_1}{\tau_2}}
      {n}
    }

    \IRULE[VPair]
    {
      \jugV{v_1}{\tau_1}{n}
      \\
      \jugV{v_2}{\tau_2}{n}
    }
    {
      \jugV
      {(v_1, v_2)}
      {\tyPROD{\tau_1}{\tau_2}}
      {n}
    }

    \IRULE[VInj$_{i \in \{ 1, 2 \}}$]
    {
      \jugV{v}{\tau_i}{n}
    }
    {
      \jugV
      {\vINJ{i}{v}}
      {\tySUM{\tau_1}{\tau_2}}
      {n}
    }

    \IRULE[VStop]
    {
      ~
    }
    {
      \jugV
      {\vSTOP}
      {\tau}
      {0}
    }

    \IRULE[VThunk]
    {
      \jugT{\Gamma}{e}{\tau}
      \\
      \jugV{\gamma}{\Gamma}{\omega}
    }
    {
      \jugV
      {\vTHUNK{e}{\gamma}}
      {\tau}
      {\omega}
    }

    \IRULE[VWarp]
    {
      \jugV
      {v}
      {\tau}
      {p(n)}
    }
    {
      \jugV
      {\vW{p}{v}}
      {\tyW{p}{\tau}}
      {n}
    }

    \fbox{\jugV{\gamma}{\Gamma}{n}}

    \IRULE
    {
      \dom{\gamma} = \dom{\Gamma}
      \\\\
      \forall x \in \dom{\gamma},
      \jugV{\gamma(x)}{\Gamma(x)}{n}
    }
    {
      \jugV{\gamma}{\Gamma}{n}
    }
  \end{mathpar}
  \caption{Typing Judgment for Values and Environments}
  \label{fig:values}
\end{figure}

The judgment~$\jugV{v}{\tau}{n}$ expresses that a value~$v$ is a prefix of some
infinite object of type~$\tau$ at~$n \in \NI$.
Its rules are given in~\Reffig{values}.
For instance,~if~$\tau$ is of the form~$\tyS{\tau'}$,~the number of elements of
type~$\tau'$ contained in~$v$ is exactly~$n$.

Closures, pairs, and injections are unremarkable.
Stream prefixes~$\vC{v_1}{v_2}$ can only be well-typed at some~$n > 0$, in which
case~$v_2$ is well-typed at~$n-1$.
The dummy value~$\vSTOP$ inhabits all types but only at~$0$.
Thunks~$\vTHUNK{e}{\gamma}$ inhabit types only at~$\omega$.
Finally, warped values~$\vW{p}{v}$ inhabit the warping modality, marking
that~$v$ has been computed at~$p(n)$.

An environment~$\gamma$ has type~$\Gamma$ at~$n$ if all its constituent values
have types at~$n$ matching those prescribed by~$\Gamma$.
%

\subsection{Evaluation Judgment}

The evaluation relation depends on several auxiliary judgments, which depend on
evaluation themselves.
They are all parameterized by a step~$n \in \NI$.
Several of these judgments have to be extended from values to environments
pointwise; since this extension is always completely unremarkable, we omit the
corresponding rules.

\paragraph{Truncation}

\begin{figure}
  \begin{mathpar}
    \fbox{\jugTR{v}{n}{v'}}

    \IRULE
    {
      ~
    }
    {
      \jugTR{s}{n+1}{s}
    }

    \IRULE
    {
      \jugTR{v_1}{n+1}{v_1'}
      \\
      \jugTR{v_2}{n}{v_2'}
    }
    {
      \jugTR
      {\vC{v_1}{v_2}}
      {n+1}
      {\vC{v_1'}{v_2'}}
    }

    \IRULE
    {
      \jugTR{\gamma}{n+1}{\gamma'}
    }
    {
      \jugTR
      {\vCLO{x}{e}{\gamma}}
      {n+1}
      {\vCLO{x}{e}{\gamma'}}
    }

    \IRULE
    {
      \jugTR{v_1}{n+1}{v_1'}
      \\
      \jugTR{v_2}{n+1}{v_2'}
    }
    {
      \jugTR
      {(v_1,v_2)}
      {n+1}
      {(v_1',v_2')}
    }

    \IRULE
    {
      \jugTR{v}{n+1}{v'}
    }
    {
      \jugTR
      {\vINJ{i}{v}}
      {n+1}
      {\vINJ{i}{v'}}
    }

    \IRULE
    {
      \jugTR{\gamma}{n+1}{\gamma'}
      \\
      \jugEV{e}{\gamma'}{n+1}{v}
    }
    {
      \jugTR{\vTHUNK{e}{\gamma}}{n+1}{v}
    }

    \IRULE
    {
      \jugTR{v}{p(n+1)}{v'}
    }
    {
      \jugTR{\vW{p}{v}}{n+1}{\vW{p}{v'}}
    }

    \IRULE
    {
      ~
    }
    {
      \jugTR{v}{0}{\vSTOP}
    }

    \IRULE
    {
      ~
    }
    {
      \jugTR{v}{\omega}{v}
    }
  \end{mathpar}
  \caption{Truncation of Values}
  \label{fig:truncation}
\end{figure}

The value typing judgment is not monotonic, in the sense
that~$\jugV{v}{\tau}{n+1}$ does not entail~$\jugV{v}{\tau}{n}$ in general.
This choice makes value typing more precise, making sure that the result of a
program of type~$\tyS{\tyI}$ at~$n$ is exactly a list containing~$n$ elements.
However, evaluation sometimes needs to turn a value at~$m$ into a value
at~$n < m$ in order to mediate between different steps.
Thus, we introduce a~\emph{truncation} judgment~$\jugTR{v}{n}{v'}$ expressing
that removing all information pertaining to steps greater than~$n$ from the
value~$v$ gives a value~$v'$.
Its rules are given in~\Reffig{truncation}.

Most rules apply when~$v$ is to be truncated to a step of the form~$n+1$.
Scalars contain the same amount of information at all finite steps, and thus
remain themselves.
The tail~$v_2$ of a stream constructor~$\vC{v_1}{v_2}$ is truncated to~$n$,
ensuring that the final stream contains~$n+1$ elements.
Closures, pairs, and injections are truncated structurally; for closures, we
truncate the environment.
To truncate a thunk to a positive finite step is to evaluate it, obtaining a
finite result; this is why truncation depends on evaluation, defined below.
To truncate a value warped by~$p$ at~$n$, truncate it at~$p(n)$.

Finally, truncation to~$0$ and truncation to~$\omega$ are symmetric.
Truncation to~$0$ erases the value completely, leaving only~$\vSTOP$.
Truncation to~$\omega$ keeps the value completely intact.

\paragraph{Coercion Application}

\begin{figure*}
  \begin{mathpar}
    \fbox{\jugEVC{\alpha}{v}{n}{v'}}

    \IRULE[ECId]
    {
      ~
    }
    {
      \jugEVC
      {\coeID}
      {v}
      {n+1}
      {v}
    }

    \IRULE[ECSeq]
    {
      \jugEVC{\alpha_1}{v_1}{n+1}{v_2}
      \\
      \jugEVC{\alpha_2}{v_2}{n+1}{v_3}
    }
    {
      \jugEVC
      {(\alpha_1;\alpha_2)}
      {v_1}
      {n+1}
      {v_3}
    }

    \IRULE[ECStream]
    {
      \jugEVC{\alpha}{v_1}{n+1}{v_1'}
      \\
      \jugEVC{\coeS{\alpha}}{v_2}{n}{v_2'}
    }
    {
      \jugEVC
      {(\coeS{\alpha})}
      {(\vC{v_1}{v_2})}
      {n+1}
      {(\vC{v_1'}{v_2'})}
    }

    \IRULE[ECFun]
    {
      ~
    }
    {
      \jugEVC
      {(\coeARR{\alpha_1}{\alpha_2})}
      {\vCLO{x}{e}{\gamma}}
      {n+1}
      {\vCLO{x}{\synCOEL{(e~(\synCOER{x}{\alpha_1}))}{\alpha_2}}{\gamma}}
    }

    \IRULE[ECProd]
    {
      \jugEVC{\alpha_1}{v_1}{n+1}{v_1'}
      \\
      \jugEVC{\alpha_2}{v_2}{n+1}{v_2'}
    }
    {
      \jugEVC
      {(\coePROD{\alpha_1}{\alpha_2})}
      {(v_1, v_2)}
      {n+1}
      {(v_1', v_2')}
    }

    \IRULE[ECSum]
    {
      \jugEVC{\alpha_i}{v}{n+1}{v'}
    }
    {
      \jugEVC
      {(\coeSUM{\alpha_1}{\alpha_2})}
      {\vINJ{i}{v}}
      {n+1}
      {\vINJ{i}{v'}}
    }

    \IRULE[ECWarp]
    {
      \jugEVC{\alpha}{v}{p(n+1)}{v'}
    }
    {
      \jugEVC
      {(\coeW{p}{\alpha})}
      {\vW{p}{v}}
      {n+1}
      {\vW{p}{v'}}
    }

    \IRULE[ECWrap]
    {
      ~
    }
    {
      \jugEVC
      {\coeWRAP}
      {v}
      {n+1}
      {\vW{\wID}{v}}
    }

    \IRULE[ECUnwrap]
    {
      ~
    }
    {
      \jugEVC
      {\coeUNWRAP}
      {\vW{\wID}{v}}
      {n+1}
      {v}
    }

    \IRULE[ECConcat]
    {
      ~
    }
    {
      \jugEVC
      {\coeCONCAT{p}{q}}
      {\vW{p}{\vW{q}{v}}}
      {n+1}
      {\vW{p \opON q}{v}}
    }

    \IRULE[ECConcatStop]
    {
      ~
    }
    {
      \jugEVC
      {\coeCONCAT{p}{q}}
      {\vW{p}{\vSTOP}}
      {n+1}
      {\vW{p \opON q}{\vSTOP}}
    }

    \IRULE[ECConcatBox]
    {
      ~
    }
    {
      \jugEVC
      {\coeCONCAT{p}{q}}
      {\vW{p}{\vTHUNK{e}{\gamma}}}
      {n+1}
      {\vW{p \opON q}{\vTHUNK{e}{\gamma}}}
    }

    \IRULE[ECDecat]
    {
      ~
    }
    {
      \jugEVC
      {\coeDECAT{p}{q}}
      {\vW{p \opON q}{v}}
      {n+1}
      {\vW{p}{\vW{q}{v}}}
    }

    \IRULE[ECDist]
    {
      ~
    }
    {
      \jugEVC
      {\coeDISTP}
      {\vW{p}{(v_1,v_2)}}
      {n+1}
      {(\vW{p}{v_1},\vW{p}{v_2})}
    }

    \IRULE[ECDistStop]
    {
      ~
    }
    {
      \jugEVC
      {\coeDISTP}
      {\vW{p}{\vSTOP}}
      {n+1}
      {(\vW{p}{\vSTOP},\vW{p}{\vSTOP})}
    }

    \IRULE[ECDistBox]
    {
      ~
    }
    {
      \jugEVC
      {\coeDISTP}
      {\vW{p}{\vTHUNK{e}{\gamma}}}
      {n+1}
      {(\vW{p}{\vTHUNK{e}{\gamma}},\vW{p}{\vTHUNK{e}{\gamma}})}
    }

    \IRULE[ECFact]
    {
      ~
    }
    {
      \jugEVC
      {\coeFACTP}
      {(\vW{p}{v_1},\vW{p}{v_2})}
      {n+1}
      {\vW{p}{(v_1,v_2)}}
    }

    \IRULE[ECInfl]
    {
      ~
    }
    {
      \jugEVC
      {\coeINFL}
      {c}
      {n+1}
      {\vW{\wOM}{\vTHUNK{c}{\emptyset}}}
    }

    \IRULE[ECDelay]
    {
      \jugTR
      {v}
      {q(n+1)}
      {v'}
    }
    {
      \jugEVC
      {\coeDEL{p}{q}}
      {v}
      {n+1}
      {v'}
    }

    \IRULE[ECZero]
    {
      ~
    }
    {
      \jugEVC{\alpha}{v}{0}{\vSTOP}
    }

    \IRULE[ECOmega]
    {
      ~
    }
    {
      \jugEVC
      {\alpha}
      {\vTHUNK{e}{\gamma}}
      {\omega}
      {\vTHUNK{\synCOER{e}{\alpha}}{\gamma}}
    }
  \end{mathpar}
  \caption{Coercion Application Judgment}
  \label{fig:coercion-application}
\end{figure*}



The judgment~$\jugEVC{\alpha}{v}{n}{v'}$ expresses that~$v'$ is the result of
coercing~$v$ by~$\alpha$.
Its rules are given in~\Reffig{coercion-application}.

As for truncation, most rules here deal with finite positive~$n$.
The identity coercion does nothing,~$\alpha_1;\alpha_2$ first applies~$\alpha_1$
then~$\alpha_2$.
The remaining composite coercions apply coercions in depth, as expected; note
that~$\coeW{p}{\alpha}$ applies~$\alpha$ at~$p(n+1)$.
The wrapping~(resp.~unwrapping) coercion adds~(resp.~removes) a
constructor~$\vW{\wID}{-}$.
The coercions~$\coeCONCAT{p}{q}$ and~$\coeDISTP$ implement the transformations
and commutations corresponding to their types, but have to deal with the cases
where~$p(n+1) = 0$ or~$p(n+1) = \omega$ explicitly.
The coercions~$\coeDECAT{p}{q}$ and~$\coeFACTP$ are similar but simpler.
Inflation creates a dummy thunk around a scalar; this is type safe since scalars
are well-typed at any finite~$n$.
A delay coercion~$\coeDEL{p}{q}$ receives an input at~$p(n+1)$ and truncates it
to~$q(n+1)$, which is smaller or equal to~$p(n+1)$ if~$\coeDEL{p}{q}$ is
well-typed.

Evaluating a coercion at~$0$ immediately returns~$\vSTOP$, as for
truncation.
On the other hand, a coercion applied at~$\omega$ is necessarily applied to a
thunk, and must be delayed itself.
We accomplish this by pushing the coercion inside the thunk.

We have elided the context-coercion application judgment, which simply lifts
coercion application componentwise to environments.
%

\paragraph{Evaluation}

\begin{figure}
  \begin{mathpar}
    \fbox{\jugEV{e}{\gamma}{n}{v}}

    \IRULE[EVar]
    {
      ~
    }
    {
      \jugEV
      {x}
      {\gamma}
      {n+1}
      {\gamma(x)}
    }

    \IRULE[EFun]
    {
      ~
    }
    {
      \jugEV
      {\synFUN{x}{\tau}{e}}
      {\gamma}
      {n+1}
      {\vCLO{x}{e}{\gamma}}
    }

    \IRULE[EApp]
    {
      \jugEV
      {e_1}
      {\gamma}
      {n+1}
      {\vCLO{x}{e}{\gamma'}}
      \\
      \jugEV
      {e_2}
      {\gamma}
      {n+1}
      {v}
      \\
      \jugEV
      {e}
      {\gamma'[x \mapsto v]}
      {n+1}
      {v'}
    }
    {
      \jugEV
      {e_1~e_2}
      {\gamma}
      {n+1}
      {v'}
    }

    \IRULE[EPair]
    {
      \jugEV
      {e_1}
      {\gamma}
      {n+1}
      {v_1}
      \\
      \jugEV
      {e_2}
      {\gamma}
      {n+1}
      {v_2}
    }
    {
      \jugEV
      {(e_1, e_2)}
      {\gamma}
      {n+1}
      {(v_1, v_2)}
    }

    \IRULE[EProj$_{i \in \{1, 2\}}$]
    {
      \jugEV
      {e}
      {\gamma}
      {n+1}
      {(v_1, v_2)}
    }
    {
      \jugEV
      {\synPROJ{i}{e}}
      {\gamma}
      {n+1}
      {v_i}
    }

    \IRULE[EInj$_{i \in \{1, 2\}}$]
    {
      \jugEV
      {e}
      {\gamma}
      {n+1}
      {v}
    }
    {
      \jugEV
      {\synINJ{i}{\tau}{e}}
      {\gamma}
      {n+1}
      {\vINJ{i}{v}}
    }

    \IRULE[ECase$_{i \in \{1, 2\}}$]
    {
      \jugEV
      {e}
      {\gamma}
      {n+1}
      {\vINJ{i}{v}}
      \\
      \jugEV
      {e_i}
      {\gamma[x_i \mapsto v]}
      {n+1}
      {v'}
    }
    {
      \jugEV
      {\synCASE{e}{x_1}{e_1}{x_2}{e_2}}
      {\gamma}
      {n+1}
      {v'}
    }

    \IRULE[EConst]
    {
      ~
    }
    {
      \jugEV
      {c}
      {\gamma}
      {n+1}
      {c}
    }

    \IRULE[ERec]
    {
      \jugITER
      {x}
      {e}
      {\gamma}
      {\vSTOP}
      {0}
      {n+1}
      {v}
    }
    {
      \jugEV
      {\synREC{x}{\tau}{e}}
      {\gamma}
      {n+1}
      {v}
    }

    \IRULE[EBy]
    {
      \jugEV
      {e}
      {\purge{\gamma}}
      {p(n+1)}
      {v}
    }
    {
      \jugEV
      {\synBY{e}{p}}
      {\gamma}
      {n+1}
      {\vW{p}{v}}
    }

    \IRULE[EHead]
    {
      \jugEV
      {e}
      {\gamma}
      {n+1}
      {\vC{v_1}{v_2}}
    }
    {
      \jugEV
      {\synHEAD{e}}
      {\gamma}
      {n+1}
      {v_1}
    }

    \IRULE[ETail]
    {
      \jugEV
      {e}
      {\gamma}
      {n+1}
      {\vC{v_1}{v_2}}
    }
    {
      \jugEV
      {\synTAIL{e}}
      {\gamma}
      {n+1}
      {\vW{\wPRED}{v_2}}
    }

    \IRULE[ECons]
    {
      \jugEV
      {e_1}
      {\gamma}
      {n+1}
      {v_1}
      \\
      \jugEV
      {e_2}
      {\gamma}
      {n+1}
      {\vW{\wPRED}{v_2}}
    }
    {
      \jugEV
      {\synCONS{e_1}{e_2}}
      {\gamma}
      {n+1}
      {\vC{v_1}{v_2}}
    }

    \IRULE[ECoeR]
    {
      \jugEV
      {e}
      {\gamma}
      {n+1}
      {v}
      \\
      \jugEVC
      {\alpha}
      {v}
      {n+1}
      {v'}
    }
    {
      \jugEV
      {\synCOER{e}{\alpha}}
      {\gamma}
      {n+1}
      {v'}
    }

    \IRULE[ECoeL]
    {
      \jugEVC
      {\beta}
      {\gamma}
      {n+1}
      {\gamma'}
      \\
      \jugEV
      {e}
      {\gamma'}
      {n+1}
      {v}
    }
    {
      \jugEV
      {(\synCOEL{\beta}{e})}
      {\gamma}
      {n+1}
      {v}
    }

    \IRULE[EZero]
    {
      ~
    }
    {
      \jugEV
      {e}
      {\gamma}
      {0}
      {\vSTOP}
    }

    \IRULE[EOmega]
    {
      ~
    }
    {
      \jugEV
      {e}
      {\gamma}
      {\omega}
      {\vTHUNK{e}{\gamma}}
    }
  \end{mathpar}
  \caption{Evaluation Judgment}
  \label{fig:evaluation}
\end{figure}

\begin{figure}
  \begin{mathpar}
    \fbox{\scriptsize \jugITER{x}{e}{\gamma}{v}{m}{n}{v'}}

    \IRULE[IFinish]
    {
      ~
    }
    {
      \jugITER
      {x}
      {e}
      {\gamma}
      {v}
      {n}
      {n}
      {v}
    }

    \IRULE[IStep]
    {
      m < n
      \\
      \jugTR
      {\gamma}
      {m+1}
      {\gamma'}
      \\\\
      \jugEV
      {e}
      {\gamma'[x \mapsto \vW{\wPRED}{v}]}
      {m+1}
      {v'}
      \\\\
      \jugITER
      {x}
      {e}
      {\gamma}
      {v'}
      {m+1}
      {n}
      {v''}
    }
    {
      \jugITER
      {x}
      {e}
      {\gamma}
      {v}
      {m}
      {n}
      {v''}
    }
  \end{mathpar}
  \caption{Iteration Judgment}
  \label{fig:iteration}
\end{figure}

The evaluation judgment is given in~\Reffig{evaluation}.

Again, most of the work is done at~$0 < n < \omega$, so we begin by decribing
the corresponding rules.
The rules for variables, functions, application, pairs, projections, injections,
pattern-matching, and scalars are the standard ones of
call-by-value~$\lambda$-calculus.
We will explain recursion shortly.
To evaluate~$\synBY{e}{p}$ at~$n+1$, evaluate~$e$ at~$p(n+1)$.
Its result should be wrapped in~$\vW{p}{-}$ to mark its provenance, and
symmetrically the environment~$\gamma$ should be purged of a layer
of~$\vW{p}{-}$ value formers.
The latter operation is denoted by~$\purge{\gamma}$; it also removes
bindings~$(x, v)$ where~$v$ is not of the form~$\vW{p}{v}$ from~$\gamma$.
Coercions rely on the coercion application judgment and its lifting to context
coercions.

All terms evaluate to~$\vSTOP$ at~$0$.
The evaluation of a term~$e$ at~$\omega$ suspends its execution, building a
thunk~$\vTHUNK{e}{\gamma}$ pairing it with the current environment~$\gamma$.

\paragraph{Recursion and Iteration}

Rule~\rn{ERec} depends on the~\emph{iteration}
judgment~$\jugITER{x}{e}{\gamma}{v}{m}{n}{v'}$.
To explain this judgment informally, let us write~$f$ for~$\synFUN{x}{\_}{e}$
and assume that~$m \le n$.
Then, this judgment computes~$v' = f^{n-m}(v)$.
Its use in~\Refrule{ERec} with~$m = 0$ and~$v = \vSTOP$ ensures
that~$v = f^n(\vSTOP)$.
Thus iteration can be viewed as an operational approximation of Kleene's
fixpoint theorem if one identifies~$\vSTOP$ with~$\bot$ from domain theory.

Rule~\rn{IFinish} terminates the iteration sequence when~$m = n$.
Rule~\rn{IStep} computes~$f(f^m(\vSTOP)) = f^{m+1}(\vSTOP)$ if~$m < n$.
The environment~$\gamma$ is well-typed at~$n$ and so must be truncated to~$m+1$.

\subsection{Metatheoretical Results}

Our evaluation judgments represent runtime errors by the absence of a result, as
is common in big-step semantics.
Thus, our judgments define partial functions:~there is at most one value~$v$
such that~$\jugEV{e}{\gamma}{n}{v}$, and similarly for all the other judgments.


\paragraph{Type Safety}

The following basic type safety result ensures that if a closed program of
type~$\tau$ evaluates to a value~$v$ at~$n$, then~$v$ is of type~$\tau$ at~$n$.
Given the typing rules for values, this ensure in particular that streams have
the length described by their types.

\begin{theorem}[Type Safety]
  If~$\jugT{\Gamma}{e}{\tau}$,~$\jugV{\gamma}{\Gamma}{n}$,
  and~$\jugEV{e}{\gamma}{n}{v}$, then~$\jugV{v}{\tau}{n}$.
\end{theorem}

Since the evaluation judgment depends on the truncation, coercion application,
and iteration judgments and vice-versa, the proof must proceed by mutual
induction, using the relevant type safety lemmas for auxiliary judgments.

\begin{lemma}[Type Safety, Truncation]
  \label{lemma:type-safety-truncation}
  If~$\jugV{v}{\tau}{m}$ and~$\jugTR{v}{n}{v'}$ with~$n \le m$,
  then~$\jugV{v'}{\tau}{n}$.
\end{lemma}

\begin{lemma}[Type Safety, Coercion Application]
  If~$\jugC{\alpha}{\tau}{\tau'}$,~$\jugV{v}{\tau}{n}$,
  and~$\jugEVC{\alpha}{v}{n}{v'}$, then~$\jugV{v'}{\tau'}{n}$.
\end{lemma}

\begin{lemma}[Type Safety, Iteration]
  If~$\jugT{\Gamma, x :
    \tyW{\wPRED}{\tau}}{e}{\tau}$,~$\jugV{\gamma}{\Gamma}{n}$,~$\jugV{v}{\tau}{m}$
  and~$\jugITER{x}{e}{\gamma}{v}{m}{n}{v'}$, then~$\jugV{v}{\tau}{n}$.
\end{lemma}

\paragraph{Totality}

In addition to the usual type errors, in our setting partiality may also arise
from time-related operations.
For instance, a term might try to truncate a value at~$n$ to~$m > n$, or to
evaluate~$\synBY{e}{p}$ in an environment which contains values that are not of
the form~$\vW{p}{-}$.
The following theorem asserts that this cannot occur with well-typed terms.

\begin{theorem}[Totality]
  \label{thm:totality}
  If~$\jugT{\Gamma}{e}{\tau}$ and~$\jugV{\gamma}{\Gamma}{n}$, then there
  exists~$v$ such that~$\jugEV{e}{\gamma}{n}{v}$.
\end{theorem}

The proof uses a realizability predicate, as explained in~\Refapp{proofs}.
It also requires the following result, also used
in~\Refsec{denotational-semantics}.

\begin{prop}[Functoriality of Truncation]
  \label{prop:functoriality-of-truncation}
  If~$\jugTR{v}{n}{v'}$ and~$\jugTR{v}{m}{v''}$ with~$m \le n$,
  then~$\jugTR{v'}{m}{v''}$.
\end{prop}

\paragraph{Monotonicity}

Finally, we prove that evaluation indeed computes longer and longer prefixes of
the same object.

\begin{prop}[Monotonicity]
  \label{prop:monotonicity}
  If~$\jugEV{e}{\gamma}{n}{v}$ and~$\jugEV{e}{\gamma'}{m}{v'}$ with~$m \le n$
  and~$\jugTR{\gamma}{m}{\gamma'}$, then~$\jugTR{v}{m}{v'}$.
\end{prop}

\paragraph{Coherence}

We have defined evaluation only on explicit terms, and indeed coercions play a
crucial role in determining the result of a computation.
Thus the question of~\emph{coherence} arises:~do all refiners of the same
implicit term having the same type compute the same results?
We give a positive answer to this question in~\Refsec{algorithmic-type-checking}
using the denotational semantics developed in the next section.

\section{Denotational Semantics}
\label{sec:denotational-semantics}

\subsection{Preliminaries}

Let~$\Bool$ denote the category with two objects~$\bot,\top$ and a single
non-identity morphism~$t : \bot \to \top$, equipped with the strict monoidal
structure given by conjunction.
Then, preorders are~$\Bool$-enriched categories, and~$\Psh{P}$ is
isomorphic to~$\catop{P} \to \Bool$.
Let~$\catpo$ denote the strict monoidal functor~$\catpo : \Bool \to \Set$
sending~$\bot$ to~$\emptyset$ and~$\top$ to~$\{ * \}$.
We write~$\catpo[P]$ for the degenerate category associated with a preorder~$P$;
its hom-sets contain at most one morphism.

Given a category~$\catC$, we denote by~$\Psh{\catC} = \catop{\catC} \to \Set$
the category of contravariant presheaves over~$\catC$.
Note that~$\Psh{P}$ differs from~$\Psh{\catpo[P]}$, hence our unusal choice of
to formally distinguish preorders from ordinary categories.

\subsection{The Topos of Trees}

In this section we sketch a model of~\Core{} in the topos of trees.
\citet{BirkedalMogelbergSchwinghammerStovring-2012} show that this category is a
convenient setting for modeling guarded recursion and synthetic step-indexing.
We follow their terminology and notations.

\begin{defn}
  The~\emph{topos of trees}, denoted~$\TT$, is~$\Psh{\catpo[\NP]}$.
\end{defn}

Briefly, an object~$X$ in the topos of trees can be described as a family of
sets~$(X(n))_{n \in \NP}$, together with a family of~\emph{restriction}
functions~$(r^X_n : X(n+1) \to X(n))_{n \in \NP}$.
The set~$X(n)$ describes what can be observed of~$X$ at step~$n$, and the
restriction functions define how future observations extend current ones.
Morphisms~$f : X \to Y$ are collections of
functions~$(f_n : X(n) \to Y(n))_{n \in \NP}$ commuting with restriction
functions.

As a topos, this category naturally has all the structure required for
interpreting simply-typed~$\lambda$-calculus with products and sums.
\begin{mathpar}
  \_ \times \_ : \TT \times \TT \to \TT
  \and
  \_ + \_ : \TT \times \TT \to \TT
  \and
  (\_)^{(\_)} : \catop{\TT} \times \TT \to \TT
\end{mathpar}
This structure follows from general constructions in presheaf categories.
Products and sums, as limits and colimits, are given pointwise.
Exponentiation can be deduced from the Yoneda lemma.

The later modality is interpreted in~$\TT$ by the functor~$\symLATER$ such
that~$(\tyLATER{X})(0) = \{ * \}$ and~$(\tyLATER{X})(n+1) = X(n)$.
A certain family of morphisms~$\semFIX{X} : X^{\tyLATER{X}} \to X$ of~$\TT$
provide fixpoint combinators, and are used to interpret guarded recursion.
We refer to~\citet{BirkedalMogelbergSchwinghammerStovring-2012} for additional
information.

\subsection{Interpreting the Warping Modality}

In order to interpret the warping modality, we need to equip the topos of trees
with a functor~$\W{p} : \TT \to \TT$ for every time warp~$p$.
Intuitively,~$(\tyW{p}{X})(n)$ should contain~``$p$-times'' more information
than~$X(n)$.
Moreover, the family of functors~$\W{(-)}$ should come equipped with enough
structure to interpret atomic coercions.

\paragraph{Pulling Presheaves along Functions}

To understand what this operation should look like, let us first consider a
restricted class of time warps.
By definition, time warps~$p$ such that~$0 < p(n) < \omega$ for
all~$0 < n < \omega$ are in a one-to-one correspondence with monotonic
functions~$f : \NP \to \NP$.
In this case, one can simply define~$(\tyW{p}{X})(n) = X(f(n))$.
Thus, if~$p$ happens to be equivalent to a function~$\NP \to \NP$,~the
functor~$\W{p} : \TT \to \TT$ is simply given by precomposition with~$p$.
From a categorical logic perspective, computing~$\tyW{p}{X}$ corresponds
to~\emph{pulling~$X$ along~$p$}.

This special case already captures some examples from the literature.
For instance,~\citet{BirkedalMogelbergSchwinghammerStovring-2012} study the left
adjoint~$\symSOONER$ of~$\symLATER$ given by~$(\tySOONER{X})(n) \defeq X(n+1)$,
which would thus correspond to~$\W{n \mapsto n+1}$.
However, most interesting time warps are not~$\NP$-valued, including those
corresponding to the later and constant modalities, and thus cannot be naively
precomposed with presheaves from~$\TT$.

\paragraph{Pulling Presheaves along Distributors}

A solution to the above problem is provided by the theory
of~\emph{distributors}, which are to functors what relations are to
functions.
A distributor~$P : \catC \catrel \catD$ from a category~$\catC$ to a
category~$\catD$ is a functor~$P : \catop{\catD} \times \catC \to \Set$.
Distributors form a~(bi)category, and enjoy properties that plain functors lack.
We refer to~\citet{Benabou-2000} for an introduction.

Any presheaf~$X : \catop{\catC} \to \Set$ is by definition equivalent to a
distributor~$\One \catrel \catC$, with~$\One$ the category with a single object
and morphism.
Postcomposing a distributor~$M : \catC \catrel \catD$ with~$X$ gives a
presheaf~$MX : \One \catrel \catD$ which, intuitively, is~\emph{$X$ pushed
  along~$M$}.
It is a crucial characteristic of distributors that post-composition with~$M$
always has a right adjoint, which we will write~$(-)/M$.
This right adjoint can be described by the end formula
\begin{equation}
  \label{eq:ran-dist}
  Y/M \defeq \int_{d \in \catD} Y(d)^{M(d,-)}.
\end{equation}
The presheaf~$Y/M$ is the result of~\emph{pulling~$Y$ along~$M$}, as recently
expounded by~\citet{MelliesZeilberger-2016}.

\paragraph{Pulling Presheaves along Time Warps}

We can extend the construction given above to time warps by realizing that the
latter are miniature distributors.

It is a consequence of the Yoneda lemma that every
distributor~$\catC \catrel \catD$ can be seen as a cocontinuous
functor~$\Psh{\catC} \to \Psh{\catD}$ and vice-versa.
A similar result holds for preorders:~every cocontinuous
function~$\Psh{P} \to \Psh{Q}$ corresponds to a monotonic
function~$\catop{Q} \times P \to \Bool$, and vice-versa.
We call such functions~\emph{linear systems}, adopting the terminology
attributed to Winskel by~\citet[\S4.1]{Hyland-2010}.
Given a time warp~$p : \Psh{\NP} \to \Psh{\NP}$, we refer to the corresponding
linear system as~$\lin{p} : \NP \catrel \NP$.
We have~$(m,n) \in \lin{p}$ if and only if~$\Yo(m) \le p(\Yo(n))$.

Pulling a presheaf along a time warp is now possible since linear systems are
nothing but~$\Bool$-enriched distributors.
Precomposing the linear system~$\lin{p}$ with~$\catpo : \Bool \to \Set$, we
obtain a standard distributor~$\catpo \lin{p} : \omega \catrel \omega$, which we
then combine with~\Refeq{ran-dist}.
\begin{defn}[Warping Functor]
  \label{defn:warping-functor}
  Given a time warp~$p$, we define the warping functor~$\W{p} : \TT \to \TT$ as
  \begin{equation}
    \W{p}
    \defeq
    (-) / (\catpo \lin{p}).
  \end{equation}
\end{defn}
Unfolding and simplifying the above definition, we obtain an explicit formula
for the observations of~$\tyW{p}{X}$ at~$n$.
\begin{equation}
  \label{eq:concrete-warping}
  (\tyW{p}{X})(n)
  =
  \left\{
    (x_m) \in
    \prod_{m = 1}^{p(n)}
    X(m)
    \;\middle|\;
    x_m = r^X_m(x_{m+1})
  \right\}
\end{equation}
The reader may check using the above formula that~$(\tyW{\wPRED}{X})(n)$
coincides with~$\tyLATER{X}(n)$.
The same is true for~$\W{\wOM}$ and~$\symCONST$.

Once the relatively intuitive~\Refeq{concrete-warping} has been found, it might
seem that the abstract~\Refdef{warping-functor} becomes unnecessary.
However, the abstract approach gives insight into the structure of~$\W{(-)}$.
In particular, routine categorical considerations imply the following.
\begin{prop}
  \label{prop:pseudo-monoidal-functor}
  Warping defines a strong monoidal functor
  \begin{equation*}
    \W{(-)}
    :
    \WARP^{\mi{op}(0,1)}
    \to
    \mi{End}(\TT).
  \end{equation*}
\end{prop}
Here~$\WARP$ and~$\mi{End}(\TT)$ are considered as monoidal categories whose
objects are respectively time warps and endofunctors of~$\TT$, and where the
monoidal structure is given by composition in both cases.
The category~$\WARP$ is preordered.
\Refprop{pseudo-monoidal-functor} entails the existence of the
following structure.
\begin{mathpar}
  \W{p \ge q} : \W{p} \to \W{q}
  \and
  \epsilon : \W{\wONE} \iso \mi{Id}
  \and
  \mu^{p,q} : \W{p} \W{q} \iso \W{p \ast q}
\end{mathpar}
Moreover, every functor~$\W{p}$ is a right adjoint, hence limit-preserving.

\subsection{The Interpretation}

Ground types are interpreted using the functor~$\Delta : \Set \to \TT$ mapping
every set to a constant presheaf.
The interpretation of~$\tyS{\tau}$ is characterized by~$
\sem{\tyS{\tau}}(n) = \prod_{m = 1}^{n} \sem{\tau}(m)
$.
We have already given the interpretation of all other types.
Typing contexts are interpreted as cartesian products.

Coercions~$\jugC{\alpha}{\tau}{\tau'}$ give rise to
morphisms~$\sem{\alpha} : \sem{\tau} \to \sem{\tau'}$.
Composite coercions are interpreted by the functorial actions of type
constructors, plus plain composition.
Atomic coercions take advantage of the structure arising
from~\Refprop{pseudo-monoidal-functor}.
For example,
$\sem{
  \jugC
  {\coeCONCAT{p}{q}}
  {\tyW{p}{\tyW{q}{\tau}}}
  {\tyW{p \opON q}{\tau}}
}
=
\mu^{p,q}_{\sem{\tau}}
$
and
$
\llbracket
  \coeDEL{p}{q} :
  \tyW{p}{\tau} \subty
  \tyW{q}{\tau}
\rrbracket
=
(\W{p \ge q})_{\sem{\tau}}
$.
The~$\coeINFL$ coercion is interpreted by the general isomorphism
between~$\Delta(S)$ and~$\tyW{\wOM}{\Delta(S)}$.
The~$\coeDISTP$ and~$\coeFACTP$ coercions are interpreted by the natural
isomorphisms arising from the limit-preservation property of~$\W{p}$.

Since the type system of~\Reffig{typing} is not exactly syntax-directed, we will
interpret typing derivations rather than terms.
Guarded recursion is interpreted using the~$\semFIX{\sem{\tau}}$ morphisms.
We interpret structure maps~$\sigma \in \structMap{\Gamma}{\Gamma'}$ as
morphisms~$\sem{\sigma} : \sem{\Gamma'} \to \sem{\Gamma}$ and~\Refrule{Struct}
by precomposition.
Other cases are standard~\cite{LambekScott-1986}.

\begin{prop}[Coherence for Explicit Terms]
  \label{prop:explicit-coherence}
  Any two derivations of~$\jugT{\Gamma}{e}{\tau}$ are interpreted by the same
  morphism in~$\TT$.
\end{prop}

The proof shows that the interpretation of any derivation of~$e$ is equal to the
interpretation of the canonical derivation for~$e$ built
in~\Refprop{canonical-derivations}.
Since this canonical derivation is unique, this entails the coherence of the
interpretation for explicit terms.

\subsection{Adequacy}

The interpretation reflects operational equivalence, which in~\Core{} consists
in observing scalars at the first step.
\begin{theorem}
  If~$\sem{\jugT{\Gamma}{e}{\tau}} = \sem{\jugT{\Gamma}{e'}{\tau}}$
  then~$\CTXEQ{\Gamma}{e}{e'}{\tau}$.
\end{theorem}
To prove the result, we remark that the values described
in~\Refsec{operational-semantics} can be organized as an object of~$\TT$, using
results such as~\Refprop{monotonicity}.
The details can be found in~\Refapp{proofs}.

\section{Algorithmic Type Checking}
\label{sec:algorithmic-type-checking}

The abstract type-checking algorithm we present in this section builds an
explicit term from an implicit one in a canonical way.
This involves two main challenges:~deciding the subtyping judgment, and dealing
with the context restriction arising in~\Refrule{Warp}.

\subsection{Deciding Subtyping}

\begin{figure*}
  \small
  \[
    \begin{array}{@{}lr@{}}
      \begin{array}[b]{@{}c@{}}
        \begin{array}{@{}r@{}l@{}}
          \begin{array}{r@{\;}l}
            \NORM{\nu}
            & = \tyW{\wOM}{\nu}
            \\
            \NORM{\tyS{\tau}}
            & =
            \tyW{\wID}{\tyS{\NORM{\tau}}}
            \\
            \NORM{\tyARR{\tau_1}{\tau_2}}
            & =
            \tyW{\wID}{(\tyARR{\NORM{\tau_1}}{\NORM{\tau_2}})}
            \\
            \NORM{\tyPROD{\tau_1}{\tau_2}}
            & =
            \tyPROD{\NORM{\tau_1}}{\NORM{\tau_2}}
          \end{array}
          &
          \begin{array}{r@{\;}l}
            \NORM{\tySUM{\tau_1}{\tau_2}}
            & =
            \tyW{\wID}{(\tySUM{\NORM{\tau_1}}{\NORM{\tau_2}})}
            \\
            \NORM{\tyW{p}{(\tyPROD{\tau_1}{\tau_2})}}
            & =
            \tyPROD{\NORM{\tyW{p}{\tau_1}}}{\NORM{\tyW{p}{\tau_2}}}
            \\
            \NORM{\tyW{p}{\tau}}
            & =
            \tyW{p \opON q}{\tau'}
            \begin{array}[t]{@{~}l}
              \mbox{where }
              \tau \ne (\tyPROD{\_}{\_})
              \\
              \mbox{and }
              \tyW{q}{\tau'} = \NORM{\tau}
            \end{array}
          \end{array}
        \end{array}
        \\
        \\
        \mbox{\normalsize (A) Type Normalization}
      \end{array}
      &
      \begin{array}[b]{@{}c@{}}
        \begin{array}{@{}r@{\;}c@{\;}l@{}}
          \Prec{\nu}{\nu}
          & = &
          \coeID
          \\
          \Prec
          {\tyS{\tau_1}}
          {\tyS{\tau_2}}
          & = &
          \coeS{\Prec{\tau_1}{\tau_2}}
          \\
          \Prec
          {\tyARR{\tau_1'}{\tau_1''}}
          {\tyARR{\tau_2'}{\tau_2''}}
          & = &
          \coeARR
          {\Prec{\tau_2'}{\tau_1'}}
          {\Prec{\tau_1''}{\tau_2''}}
          \\
          \Prec
          {\tyPROD{\tau_1'}{\tau_1''}}
          {\tyPROD{\tau_2'}{\tau_2''}}
          & = &
          \coePROD
          {\Prec{\tau_1'}{\tau_2'}}
          {\Prec{\tau_1''}{\tau_2''}}
          \\
          \Prec
          {\tySUM{\tau_1'}{\tau_1''}}
          {\tySUM{\tau_2'}{\tau_2''}}
          & = &
          \coeSUM
          {\Prec{\tau_1'}{\tau_2'}}
          {\Prec{\tau_1''}{\tau_2''}}
          \\
          \Prec
          {\tyW{p}{\tau_1}}
          {\tyW{q}{\tau_2}}
          & = &
          \coeDEL{p}{q}
          ;
          \coeW{q}{\Prec{\tau_1}{\tau_2}}
          \mbox{ if } p \ge q
        \end{array}
        \\
        \\
        \mbox{\normalsize (B) Type Precedence}
      \end{array}
    \end{array}
  \]
  \caption{Type Normalization and Precedence}
  \label{fig:alg-types}
\end{figure*}

To decide subtyping, we start with the observation that most atomic
coercions~$\jugC{\alpha}{\tau}{\tau'}$ from~\Reffig{subtyping} come in pairs, in
the sense that there exists~$\alpha^{-1}$ such
that~$\jugC{\alpha^{-1}}{\tau'}{\tau}$.
This is even true for~$\coeINFL$, since we can take~$\coeINFL^{-1}$ to
be~$\coeDEL{\wOM}{\wID};\coeUNWRAP$.
The only atomic coercion for which this is not the case is~$\coeDEL{p}{q}$
when~$q < p$.
This suggests dealing with delays separately.

\paragraph{Normalizing Types}

To deal with invertible coercions, we define a function~$\tau$ mapping each type
to an equivalent but simpler form.
Such~\emph{normal} types~$\tau^n$ obey the following grammar.
{
  \small
  \[
    \begin{array}{@{}l@{\hspace{-.1pt}}r@{}}
      \begin{array}{l@{~}c@{~}l}
        \tau^n
        & \Coloneqq &
        \tyW{p}{\tau^r} \mid \tyPROD{\tau^n}{\tau^n}
      \end{array}
      &
      \begin{array}{l@{~}c@{~}l}
        \tau^r
        & \Coloneqq &
        \nu
        \mid \tyS{\tau^n}
        \mid \tyARR{\tau^n}{\tau^n}
        \mid \tySUM{\tau^n}{\tau^n}
      \end{array}
    \end{array}
  \]
}
In other words, normal types feature exactly one warping modality immediately
above every non-product type former.

The total function~$\NORM{\tau}$ returns the normalized form of~$\tau$.
It is defined by recursion on~$\tau$ in~\Reffig{alg-types}-A.
For every~$\tau$, there are
coercions~$\jugCE{\NORMin{\tau}}{\NORMout{\tau}}{\tau}{\NORM{\tau}}$, defined
in~\Refapp{supplementary-material}.

\paragraph{Deciding Precedence}

We now decide subtyping in the special case where the only atomic coercions
allowed are delays, a case we call~\emph{precedence}.
The corresponding partial computable function~$\Prec{-}{-}$, when defined,
builds a coercion~$\jugC{\Prec{\tau}{\tau'}}{\tau}{\tau'}$.
It is given in~\Reffig{alg-types}-B.
In the absence of~$\coeCONCAT{p}{q}$ and~$\coeDECAT{p}{q}$ coercions, it is
enough to traverse~$\tau$ and~$\tau'$ in lockstep, checking whether~$p \ge q$
holds when comparing~$\tyW{p}{\tau}$ and~$\tyW{q}{\tau'}$.

\paragraph{Putting it all together}

We decide subtyping in the general case by combining precedence with
normalization:
\[
  \Coe{\tau}{\tau'}
  \defeq
  \NORMin{\tau}
  ;
  \Prec{\NORM{\tau}}{\NORM{\tau'}}
  ;
  \NORMout{\tau'}.
\]
We write~$\Coe{\Gamma}{\Gamma'}$ for the pointwise extension to contexts.

\subsection{Adjoint Typing Contexts}

Consider the type-checking problem for~$\synBY{t}{p}$ in a given
context~$\Gamma$.
If every type~$\tau = \Gamma(x)$, with~$x$ a free variable of~$t$, is of the
form~$\tyW{p}{\tau'}$, we may apply~\Refrule{Warp}.
Otherwise, we have to find a type~$\tau'$ such
that~$\tau \subty \tyW{p}{\tau'}$.
There are several choices for~$\tau'$, and they are far from equivalent.
For instance, taking~$\tau' \defeq \tyW{\wZERO}{\tau}$ would work since
$
  \tau
  \subty \tyW{\wZERO}{\tau}
  \subty \tyW{p}{\tyW{\wZERO}{\tau}}
$ always holds,
but will in general impose artificial constraints on the type
of~$t$.
For~$\tau'$ to be a canonical choice,
\begin{equation}
  \label{eq:type-adjoint1}
  \tau
  \subty
  \tyW{p}{\tau''}
  \Leftrightarrow
  \tau'
  \subty
  \tau''
\end{equation}
needs to hold for any type~$\tau''$.
Now, assume that~$\tau$ and~$\tau''$ are normalized types which are not
products, and thus necessarily start with a warping modality.
Equivalence~\Refeqshort{type-adjoint1} becomes
\begin{equation}
  \label{eq:type-adjoint2}
  \tyW{q}{\tau}
  \subty
  \tyW{p}{\tyW{r}{\tau''}}
  \Leftrightarrow
  \tau'
  \subty
  \tyW{r}{\tau''}.
\end{equation}
Then, a solution satisfying~\Refeqshort{type-adjoint1} is given
by~$\tau' = \tyW{q \wD p}{\tau}$, with~$q \wD p$ a hypothetical time warp such
that
\begin{equation}
  \label{eq:warp-adjoint}
  r \circ p \le q
  \Leftrightarrow
  r \le q \wD p.
\end{equation}
We are thus looking for an operation~$(-) \wD p$ right adjoint to
precomposition~$(-) \circ p$.
Right adjoints to precomposition~(and postcomposition,
cf.~\Refsec{denotational-semantics}) always exist for
distributors~\citep[\S4]{Benabou-2000}, and thus linear systems.
The general formula, specialized to linear systems and time warps, gives
\begin{equation}
  (q \wD p)(n)
  =
  p(
  \min
  \{
  m \in \NI
  \mid
  n \le q(m)
  \}).
\end{equation}
Thus, we define normal-type division as
\[
  (\tyPROD{\tau_1}{\tau_2}) \wDn p = \tyPROD{(\tau_1 \wDn p)}{(\tau_2 \wDn p)}
  \mbox{ and }
  (\tyW{q}{\tau}) \wDn p = \tyW{q \wD p}{\tau}
\]
and general type division as~$\tau \wD p \defeq \NORM{\tau} \wDn p$.
%

\subsection{The Algorithm}

\begin{figure*}
  \centering
  \small
  $\begin{array}{@{}r@{\;}c@{\;}l@{}}
    \Elab
    {\Gamma}
    {x}
    & = &
    (
    \Gamma(x)
    ,
    x
    )
    \\
    \Elab
    {\Gamma}
    {\synFUN{x}{\tau}{t}}
    & = &
    (
    \tyARR{\tau}{\tau'}
    ,
    \synFUN{x}{\tau}{e}
    )
    \mbox{ where }
    (\tau', e) = \Elab{\Gamma, x : \tau}{t}
    \\
    \Elab
    {\Gamma}
    {t_1~t_2}
    & = &
    (
    \tau_1''
    ,
    (e_1; \Coe{\tau_1}{\tyARR{\tau_1'}{\tau_1''}})
    ~
    (e_2; \Coe{\tau_2}{\tau_1'})
    )
    \\ & &
    \mbox{where }
    (\tau_i, e_i) = \Elab{\Gamma}{t_i}
    \mbox{ and }
    \tyW{-}{(\tyARR{\tau_1'}{\tau_1''})} = \NORM{\tau_1}
    \\
    \Elab
    {\Gamma}
    {(t_1, t_2)}
    & = &
    (
    \tyPROD{\tau_1}{\tau_2}
    ,
    (e_1, e_2)
    )
    \mbox{ where }
    (\tau_i, e_i) = \Elab{\Gamma}{t_i}
    \\
    \Elab
    {\Gamma}
    {\synPROJ{i}{t}}
    & = &
    (
    \tau_i
    ,
    \synPROJ{i}{(e;\Coe{\tau}{\tyPROD{\tau_1}{\tau_2}})}
    )
    \mbox{ where }
    (\tau, e) = \Elab{\Gamma}{t}
    \mbox{ and }
    \tyPROD{\tau_1}{\tau_2} = \NORM{\tau}
    \\
    \Elab
    {\Gamma}
    {\synINJ{1+i}{\tau_{2-i}}{t}}
    & = &
    (
    \tySUM{\tau_1}{\tau_2}
    ,
    \synINJ{1+i}{\tau_{2-i}}{e}
    )
    \mbox{ where }
    (\tau_{1+i}, e) = \Elab{\Gamma}{t}
    \\
    \Elab
    {\Gamma}
    {\synCASE{t}{x_1}{t_1}{x_2}{t_2}}
    & = &
    (
    \SUP{\tau_1'}{\tau_2'}
    ,
    \synCASE
    {e;\Coe{\tau}{\tySUM{\tau_1}{\tau_2}}}
    {x_1}{e_1;\Coe{\tau_1'}{\SUP{\tau_1'}{\tau_2'}}}
    {x_2}{e_2;\Coe{\tau_2'}{\SUP{\tau_1'}{\tau_2'}}}
    )
    \\ & &
    \mbox{ where }
    (\tau, e) = \Elab{\Gamma}{t}
    \mbox{ and }
    \tyW{-}{(\tySUM{\tau_1}{\tau_2})} = \NORM{\tau}
    \mbox{ and }
    (\tau_i', e_i) = \Elab{\Gamma, x : \tau_i}{t_i}
    \\
    \Elab
    {\Gamma}
    {s}
    & = &
    (\nu, s)
    \mbox{ where } s \in S_\nu
    \\
    \Elab
    {\Gamma}
    {\synREC{x}{\tau}{t}}
    & = &
    (
    \tau
    ,
    \synREC{x}{\tau}{(e;\Coe{\tau'}{\tau})}
    )
    \mbox{ where }
    (\tau', e) = \Elab{\Gamma, x : \tyW{\wPRED}{\tau}}{t}
    \\
    \Elab
    {\Gamma}
    {\synBY{t}{p}}
    & = &
    (
    \tyW{p}{\tau}
    ,
    \Coe{\Gamma}{\tyW{p}{(\Gamma \wD p)}};
    \synBY{e}{p}
    )
    \mbox{ where }
    (\tau, e) = \Elab{\Gamma \wD p}{t}
    \\
    \Elab
    {\Gamma}
    {\synHEAD{t}}
    & = &
    (\tau', \synHEAD{(e;\Coe{\tau}{\tyS{\tau'}})})
    \mbox{ where }
    (\tau, e) = \Elab{\Gamma}{t}
    \mbox{ and }
    \tyW{-}{\tyS{\tau'}} = \NORM{\tau}
    \\
    \Elab
    {\Gamma}
    {\synTAIL{t}}
    & = &
    (\tyW{\wPRED}{\tyS{\tau'}}, \synTAIL{(e;\Coe{\tau}{\tyS{\tau'}})})
    \mbox{ where }
    (\tau, e) = \Elab{\Gamma}{t}
    \mbox{ and }
    \tyW{-}{\tyS{\tau'}} = \NORM{\tau}
    \\
    \Elab
    {\Gamma}
    {\synCONS{t_1}{t_2}}
    & = &
    (
    \tyS{(\SUP{\tau_1}{\tau_2'})}
    ,
    \synCONS
    {
      (e_1;\Coe{\tau_1}{\SUP{\tau_1}{\tau_2'}})
    }
    {
      (e_2;\Coe{\tau_2}{\tyW{\wPRED}{\tyS{(\SUP{\tau_1}{\tau_2'})}}})
    }
    )
    \\ & &
    \mbox{where }
    (\tau_i, e_i) = \Elab{\Gamma}{t_i}
    \mbox{ and }
    \tyW{-}{\tyS{\tau_2'}} = \NORM{\tau_2}
  \end{array}$
  \caption{Elaboration}
  \label{fig:term-elaboration}
\end{figure*}

The partial computable function~$\Elab{\Gamma}{t}$ returns a pair~$(\tau,e)$
with~$\Refines{e}{t}$ such that~$\jugT{\Gamma}{e}{\tau}$ holds.
Its definition is given in~\Reffig{term-elaboration}.
It uses the algorithmic subtyping judgment when type-checking destructors, and
the context division judgment when applying~\Refrule{Warp}.
The case of pattern-matching relies on the existence of type suprema, which are
easy to compute structurally for normal types;
see~\Refapp{supplementary-material}.


\subsection{Metatheoretical Results}

\begin{lemma}
  \label{lemm:coercion-coherence}
  If~$\jugC{\alpha}{\tau}{\tau'}$ then~$\Coe{\tau}{\tau'}$ is defined and
  \[
    \sem{\jugC{\alpha}{\tau}{\tau'}}
    =
    \sem{\jugC{\Coe{\tau}{\tau'}}{\tau}{\tau'}}.
  \]
\end{lemma}

\begin{theorem}[Completeness of Algorithmic Typing]
  \label{thm:completeness}
  If~$\jugT{\Gamma}{e}{\tau}$, there is~$e^m,\tau^m,\alpha^m$
  with~$(\tau^m,e^m) = \Elab{\Gamma}{t}$,~$\jugC{\alpha^m}{\tau^m}{\tau}$,
  and
  \[
    \sem{\jugT{\Gamma}{e}{\tau}}
    =
    \sem{\jugT{\Gamma}{e}{\tau^m}}
    ;
    \sem{\jugC{\alpha}{\tau^m}{\tau}}
    .
  \]
\end{theorem}

The fact that algorithmic subtyping is deterministic together
with~\Reflemm{coercion-coherence} and~\Refthm{completeness} immediately entails
coherence.

\begin{coro}[Denotational Coherence]
  For any~$\Refines{e_1,e_2}{t}$ such that~$\jugT{\Gamma}{e_1}{\tau}$
  and~$\jugT{\Gamma}{e_2}{\tau}$, we
  have~$\sem{\jugT{\Gamma}{e_1}{\tau}} = \sem{\jugT{\Gamma}{e_2}{\tau}}$.
\end{coro}

\begin{coro}[Operational Coherence]
  For any~$\Refines{e_1,e_2}{t}$ such that~$\jugT{\Gamma}{e_1}{\tau}$
  and~$\jugT{\Gamma}{e_2}{\tau}$, we have~$\CTXEQ{\Gamma}{e_1}{e_2}{\tau}$.
\end{coro}

\section{Discussion and Related Work}
\label{sec:conclusion}

\subsection{Guarded Type Theories}

\paragraph{Expressiveness}

On the one hand,~\Core{} captures finer-grained temporal information than
existing guarded type theories, and also recasts their modalities in a uniform
setting.
We illustrate this point by comparing~\Core{} to
the~\GLC{}~\cite{CloustonBizjakBuggeBirkedal-2016}, since they are relatively
close.
The later and constant modality correspond respectively to~$\W{\wPRED}$
and~$\W{\wOM}$.
The~\GLC{} operations~$\kw{next} : \tyARR{\tau}{\tyLATER{\tau}}$
and~$\kw{unbox} : \tyARR{\tyCONST{\tau}}{\tau}$ correspond to the
coercions~$\coeWRAP;\coeDEL{\wID}{\wPRED}$ and~$\coeDEL{\wOM}{\wID};\coeUNWRAP$.
Erasing later modalities in the~\GLC{} happens via the term former~\kw{prev},
which restricts the context to be constant~(essentially, under~$\symCONST$);
in~\Core{}, this would arise from the implicit type
equivalence~$\tyW{\wOM}{\tyW{\wPRED}{\tau}} \equiv \tyW{\wOM \opON \wPRED}{\tau}
= \tyW{\wOM}{\tau}$.
Additionally, the introduction rule for~$\symCONST$ in the~\GLC{} is more
restrictive than~\Refrule{Warp} for~$\synBY{t}{\wOM}$, since the latter allows
the free variables of~$t$ to have types~$\tyW{p}{\tau}$ where~$p$ is
constant but not necessarily~$\wOM$.
The~\GLC{} makes~$\symLATER$ into an~``applicative
functor''~\cite{McBridePaterson-2008}, implementing only the left-to-right
direction of the type
isomorphism~$\tyW{\wPRED} (\tyARR{\tau_1}{\tau_2}) \iso
\tyARR{\tyW{\wPRED}{\tau_1}}{\tyW{\wPRED}{\tau_2}}$.
In~\Core{}, both directions are definable, the right-to-left one as
\[
  \begin{array}{@{}l}
    \synFUN
    {\mt{f}}
    {
      \tyARR
      {\tyW{\wPRED}{\tau_1}}
      {\tyW{\wPRED}{\tau_2}}
    }
    {
      (
      \synBY
      {
        (
        \synFUN
        {x}
        {\tau_1}
        {
          \synBY
          {
            (\mt{f}~\mt{x})
          }
          {
            \wSUCC
          }
        }
        )
      }
      {\wPRED}
      )
    }
  \end{array}
\]
where~$\wSUCC$ is is the time warp which is left adjoint
to~$\wPRED$~($\symSOONER$
in~\cite{BirkedalMogelbergSchwinghammerStovring-2012}).

On the other hand~\Core{} lacks many features present in other guarded type
theories~(including the~\GLC{}).
It would be useful, for instance, to replace the fixed stream type
with general guarded recursive
types~\cite{BirkedalMogelbergSchwinghammerStovring-2012,CloustonBizjakBuggeBirkedal-2016};
this requires designing a guardedness criterion in the presence of the warping
modality.
Clock variables~\cite{AtkeyMcBride-2013} would allow types to express that
unrelated program pieces may operate within disjoint time scales.
\Core{} enjoys decidable type-checking, but not type inference; in contrast,
type inference for the later modality has been studied by~\citet{Severi-2017}.
Finally,~\Core{} might be difficult to extend to dependent types, since it is
inherently call-by-value, whereas several dependent type theories with later
have been
proposed~\cite{BizjakBuggeCloustonMogelbergBirkedal-2016,BirkedalBizjakCloustonBuggeGrathwohlSpittersVezzosi-2016}.

\paragraph{Metatheory}

\Core{} also stands out among guarded type theories by the design of its
metatheory.
First, as mentioned above, its semantics fixes a call-by-value evaluation
strategy, in contrast with actual calculi enjoying
unrestricted~$\beta$-reduction.
We believe that this is natural since~$\synBY{t}{p}$ is in essence an effectful
term which modifies the current time step.

Second, the context restriction in~\Refrule{Warp} is perhaps controversial from
a technical perspective.
This kind of rule, acting on the left of the turnstile, is normally avoided in
natural-deduction presentations as it is known to
cause~``anomalies''~\cite{PfenningDavies-2001}, e.g., breaking substitution
lemmas.
Since~\Core{} is call-by-value, we do not need subtitution to hold for arbitrary
terms.
We do not expect difficulties in proving a substitution lemma for values in a
variant of~\Core{} where they have been made a subclass of expressions,
defining~$(\synBY{t}{p})[x \backslash v]$ to
be~$\synBY{t[x \backslash \purge{v}]}{p}$, with~$\purge{v}$ defined as
in~\Refsec{operational-semantics}.

Third,~\Core{} uses subtyping, which has been eschewed by guarded type theories
after Nakano's original proposal.
Yet, the context restriction of~\Refrule{Warp} makes subtyping extremely useful
in practice.
In its absence, terms would have to massage the typing context before
introducing the warping modality.
Guarded recursion would also be more difficult to use without the ability to
reason up to time warp composition.

\subsection{Synchronous Programming Languages}

\Core{} is a relative of synchronous programming languages in the vein of
Lustre~\cite[\dots]{CaspiPilaudHalbwacksPlaice-POPL-1987,CaspiPouzet-1996,CohenDurantonEisenbeisPagettiPlateau-POPL-2006,ForgetBoniolLesensPagetti-2008,Guatto-2016}.
Such languages use~``clocks''~(not to be confused with clock variables) to
describe stream growth; such a clock is a time warp whose image forms a
downward-closed subset of~$\omega$~(except in~\cite{Guatto-2016}).
Synchronous languages are generally first-order~(with
exceptions~\cite{Pouzet-2006,Guatto-2016}) and separate clock analysis from
productivity checking.
As a result,~\Core{} is both more flexible and simpler from a metatheoretical
standpoint.
However, it does not enforce bounds on space usage, in contrast with synchronous
languages or the work of
Krishnaswami~\cite{KrishnaswamiBenton-2011,Krishnaswami-2013}.

\paragraph{Acknowledgements}

This work has benefited from conversations with many researchers, including
Albert Cohen,
Louis Mandel,
Paul-André Melliès,
and Marc Pouzet.
%
It owes an especially great deal to Paul-André Melliès, who introduced the
author to the topos of trees and distributors.
This work was partially supported by the~\grantsponsor{GS501100001659}{German
  Research Council~(DFG)}{http://dx.doi.org/10.13039/501100001659} under Grant
No.~\grantnum{GS501100001659}{ME14271/6-2}.

\bibliographystyle{ACM-Reference-Format}
\bibliography{generalized-modality}

\iftoggle{fullversion}{
  \appendix
  \section{Supplementary Material}
\label{app:supplementary-material}

\subsection{Coercions to and from Normal Types}

\begin{figure}
  \begin{align*}
    \NORMin{\nu}
    & =
    \coeINFL
    \\
    \NORMin{\tyS{\tau}}
    & =
    \coeS{\NORMin{\tau}};\coeWRAP
    \\
    \NORMin{\tyARR{\tau_1}{\tau_2}}
    & =
    (\coeARR{\NORMout{\tau_1}}{\NORMin{\tau_1}});\coeWRAP
    \\
    \NORMin{\tyPROD{\tau_1}{\tau_2}}
    & =
    \coePROD{\NORMin{\tau_1}}{\NORMin{\tau_2}}
    \\
    \NORMin{\tySUM{\tau_1}{\tau_2}}
    & =
    (\coeSUM{\NORMin{\tau_1}}{\NORMin{\tau_2}});\coeWRAP
    \\
    \NORMin{\tyW{p}{(\tyPROD{\tau_1}{\tau_2})}}
    & =
    \coeDISTP;
    (
    \coePROD
    {\NORMin{\tyW{p}{\tau_1}}}
    {\NORMin{\tyW{p}{\tau_2}}}
    )
    \\
    \NORMin{\tyW{p}{\tau}}
    & =
    \coeW{p}{\NORMin{\tau}};\coeCONCAT{p}{q}
    \begin{array}[t]{@{}l}
      \mbox{ where }
      \tau \ne (\tyPROD{\_}{\_})
      \\
      \mbox{ and }
      \tyW{q}{\tau'} = \NORM{\tau}
    \end{array}
    \\
    \NORMout{\nu}
    & =
    \coeDEL{\wOM}{\wID};\coeUNWRAP
    \\
    \NORMout{\tyS{\tau}}
    & =
    \coeUNWRAP;\coeS{\NORMout{\tau}}
    \\
    \NORMout{\tyARR{\tau_1}{\tau_2}}
    & =
    \coeUNWRAP;(\coeARR{\NORMin{\tau_1}}{\NORMout{\tau_1}})
    \\
    \NORMout{\tyPROD{\tau_1}{\tau_2}}
    & =
    \coePROD{\NORMout{\tau_1}}{\NORMout{\tau_2}}
    \\
    \NORMout{\tySUM{\tau_1}{\tau_2}}
    & =
    \coeWRAP;(\coeSUM{\NORMout{\tau_1}}{\NORMout{\tau_2}})
    \\
    \NORMout{\tyW{p}{(\tyPROD{\tau_1}{\tau_2})}}
    & =
    (
    \coePROD
    {\NORMout{\tyW{p}{\tau_1}}}
    {\NORMout{\tyW{p}{\tau_2}}}
    );
    \coeFACTP
    \\
    \NORMout{\tyW{p}{\tau}}
    & =
    \coeDECAT{p}{q};
    \coeW{p}{\NORMout{\tau}}
    \begin{array}[t]{@{}l}
      \mbox{ where }
      \tau \ne (\tyPROD{\_}{\_})
      \\
      \mbox{ and }
      \tyW{q}{\tau'} = \NORM{\tau}
    \end{array}
  \end{align*}
  \caption{Coercions to and from Normal Types}
  \label{fig:alg-normalize}
\end{figure}

The coercions~$\jugCE{\NORMin{\tau}}{\NORMout{\tau}}{\tau}{\NORM{\tau}}$ are
defined in~\Reffig{alg-normalize}.
They are defined in a symmetric way, except a slightt discrepancy in the case of
ground types:~one must take~$\coeDEL{\wOM}{\wID};\coeUNWRAP$ as an inverse
to~$\coeINFL$, as mentioned in~\Refsec{calculus}.

\subsection{Type Bounds}

\begin{figure}
  \[
    \begin{array}{@{}r@{\;}c@{\;}l@{}}
      \SUPn
      {\nu}
      {\nu}
      & = &
      \nu
      \\
      \SUPn
      {(\tyS{\tau})}
      {(\tyS{\tau'})}
      & = &
      \tyS{(\SUPn{\tau}{\tau'})}
      \\
      \SUPn
      {(\tyARR{\tau_1}{\tau_2})}
      {(\tyARR{\tau_1'}{\tau_2'})}
      & = &
      \tyARR
      {(\INFn{\tau_1}{\tau_1'})}
      {(\SUPn{\tau_2}{\tau_2'})}
      \\
      \SUPn
      {(\tyPROD{\tau_1}{\tau_2})}
      {(\tyPROD{\tau_1'}{\tau_2'})}
      & = &
      \tyPROD
      {(\SUPn{\tau_1}{\tau_1'})}
      {(\SUPn{\tau_2}{\tau_2'})}
      \\
      \SUPn
      {(\tySUM{\tau_1}{\tau_2})}
      {(\tySUM{\tau_1'}{\tau_2'})}
      & = &
      \tySUM
      {(\SUPn{\tau_1}{\tau_1'})}
      {(\SUPn{\tau_2}{\tau_2'})}
      \\
      \SUPn
      {(\tyW{p}{\tau_1})}
      {(\tyW{q}{\tau_2})}
      & = &
      \tyW
      {\INF{p}{q}}
      {(\SUPn{\tau_1}{\tau_2})}
      \\
      \\
      \INFn
      {\nu}
      {\nu}
      & = &
      \nu
      \\
      \INFn
      {(\tyS{\tau})}
      {(\tyS{\tau'})}
      & = &
      \tyS{(\INFn{\tau}{\tau'})}
      \\
      \INFn
      {(\tyARR{\tau_1}{\tau_2})}
      {(\tyARR{\tau_1'}{\tau_2'})}
      & = &
      \tyARR
      {(\SUPn{\tau_1}{\tau_1'})}
      {(\INFn{\tau_2}{\tau_2'})}
      \\
      \INFn
      {(\tyPROD{\tau_1}{\tau_2})}
      {(\tyPROD{\tau_1'}{\tau_2'})}
      & = &
      \tyPROD
      {(\INFn{\tau_1}{\tau_1'})}
      {(\INFn{\tau_2}{\tau_2'})}
      \\
      \INFn
      {(\tySUM{\tau_1}{\tau_2})}
      {(\tySUM{\tau_1'}{\tau_2'})}
      & = &
      \tySUM
      {(\INFn{\tau_1}{\tau_1'})}
      {(\INFn{\tau_2}{\tau_2'})}
      \\
      \INFn
      {(\tyW{p}{\tau_1})}
      {(\tyW{q}{\tau_2})}
      & = &
      \tyW
      {\SUP{p}{q}}
      {(\INFn{\tau_1}{\tau_2})}
    \end{array}
  \]
  \caption{Type Suprema and Infima for Normal Types}
  \label{fig:alg-bounds}
\end{figure}

The type-checking and elaboration algorithm presented
in~\Reffig{term-elaboration} relies on the computation of type
suprema.
\Reffig{alg-bounds} defines such suprema and infima for normal types.
The general case is obtained simply by defining~$\SUP{\tau_1}{\tau_2}$
as~$\SUPn{\NORM{\tau_1}}{\NORM{\tau_2}}$.

\subsection{Additional Examples}

\begin{example}[Natural Numbers]
  The stream of natural numbers described in~\Refsec{introduction} can be
  implemented as follows.
  \[
    \begin{array}{@{}l@{}}
      \synREC
      {\idt{nat}}
      {\tyS{\tyI}}
      {
        \synCONS
        {0}
        {
          \synBY
          {
            (
            \idt{map}\;
            (\synFUN{\idt{x}}{\tyI}{\idt{x} + 1})\;
            \idt{nat}
            )
          }
          {\wPRED}
        }
      }
    \end{array}
  \]
  We assume again that~\idt{map} has the type obtained in~\Refex{constant}.
\end{example}

We now give the complete definition of the Thue-Morse sequence discussed
in~\Refex{thue-morse-abridged}.
Following~\citet{CloustonBizjakBuggeBirkedal-2016}, our definition proceeds in
two steps:~first the function~\idt{h}, then the sequence~\idt{tm} itself.
We assume that the language has been extended with booleans and a~\idt{not}
function.

\begin{example}
  Informally, the function~\idt{h} takes a boolean stream and intersperses it
  with its pointwise negation.
  \[
    \begin{array}{@{}l@{}}
      \kw{rec}\;(\idt{h} : \tyARR{\tyS{\tyB}}{\tyW{\pw{2}}{\tyS{\tyB}}}).
      \\
      \kw{fun}\;(\idt{xs} : \tyS{\tyB}).
      \\
      \synLLET
      {\idt{x}}
      {\tyB}
      {\synHEAD{\idt{xs}}}
      \\
      \kw{let}\;
      \idt{x} : \tyB = \synHEAD{\idt{xs}}
      \;\kw{and}\;
      \idt{ys} : \tyW{\upw{0}{1}}{\tyS{\tyB}} = \synTAIL{\idt{ys}}
      \;\kw{in}
      \\
      \synLLET
      {\idt{zs}}
      {\tyW{\upw{0}{1}}{\tyW{\pw{2}}{\tyS{\tyB}}}}
      {
        \synBY
        {(\idt{h}\;\idt{ys})}
        {\wPRED}
      }
      \\
      \synBY
      {
        (
        \idt{x}
        \symCONS
        \synBY
        {
          (
          \synCONS
          {(\idt{not}\;\idt{x})}
          {\idt{zs}}
          )
        }
        {\upw{0}{1}}
        )
      }
      {\pw{2}}
    \end{array}
  \]
  As in previous examples, the recursive call happens
  under~$\kw{by}\,\upw{0}{1}$, ensuring it does not happen at the first time
  step.
  Since~$\idt{x}$ is of a scalar type, it is in effect not subject to the
  context restriction in rule~\rn{Warp}; for example, we have
  \[
    \jugC
    {\coeINFL;\coeDEL{\upw{\omega}{0}}{\pw{2}}}
    {\tyB}
    {\tyW{\pw{2}}{\tyB}}
    .
  \]
  To see why the use of~\idt{zs} at type~$\tyW{\upw{0}{1}}{\tyS{\tyB}}$ is
  well-typed, observe that~$\upw{0}{2} \wD~\pw{2} = \upw{0\,0}{2\,0}$ and
  \[
    \begin{array}{@{}r@{\;}c@{\;}l@{}}
      \tyW{\upw{0}{1}}{\tyW{\pw{2}}{\tyS{\tyB}}}
      & \equiv &
      \tyW{\upw{0}{2}}{\tyS{\tyB}}
      \\
      & \equiv &
      \tyW{\pw{2}}{\tyW{\upw{0\,0}{2\,0}}{\tyS{\tyB}}}
      \\
      & \equiv &
      \tyW{\pw{2}}{\tyW{\upw{0}{1}}{\tyW{\upw{0}{2\,0}}{\tyS{\tyB}}}}
      \\
      & \subty &
      \tyW{\pw{2}}{\tyW{\upw{0}{1}}{\tyW{\upw{0}{1}}{\tyS{\tyB}}}}.
    \end{array}
  \]
\end{example}

In the next example, we assume that the definition of~\idt{h} given above has
been warped by~$\upw{\omega}{0}$~(i.e.,~$\wOM$), as in~\Refex{constant}.

\begin{example}
  The Thue-Morse sequence~\idt{tm} is defined below.
  \[
    \begin{array}{@{}l@{}}
      \kw{rec}\;(\idt{tm} : \tyS{\tyB}).
      \\
      \synLLET
      {\idt{tm'}}
      {\tyW{\upw{0}{1}}{\tyW{\pw{2}}{\tyS{\tyB}}}}
      {\synBY{(\idt{h}~\idt{tm})}{\upw{0}{1}}}
      \\
      \synCONS
      {\kw{false}}
      {
        \synBY{(\synTAIL{\idt{tm'}})}{\upw{0\,2}{1}}
      }
    \end{array}
  \]
  The key subterm is~$\synBY{(\synTAIL{\idt{tm'}})}{\upw{0\,2}{1}}$.
  Informally, this allows us to run~\kw{tail} twice at the second time step,
  obtaining one element out of~\idt{tm'}.
  Technically speaking, the use of~$\idt{tm'}$ with type~$\tyS{\tyB}$ is
  justified by~$\upw{0}{2} \wD~ \upw{0\,2}{1} = \upw{2\,0}{2}$ and
  \[
    \begin{array}{@{}r@{\;}c@{\;}l@{}}
      \tyW{\upw{0}{1}}{\tyW{\pw{2}}{\tyS{\tyB}}}
      & \equiv &
      \tyW{\upw{0}{2}}{\tyS{\tyB}}.
      \\
      & \equiv &
      \tyW{\upw{0\,2}{1}}{\tyW{\upw{2\,0}{2}}{\tyS{\tyB}}}
      \\
      & \subty &
      \tyW{\upw{0\,2}{1}}{\tyS{\tyB}}
    \end{array}
  \]
  where the last step is performed
  by~$\coeDEL{\upw{2\,0}{2}}{\pw{1}};\coeUNWRAP$.
  This last step shows that the productivity of this definition only depends on
  the fact that~\idt{h} can produce its first two output elements from its first
  input element.
  Indeed, this definition still type-checks when~\idt{h} is given the strictly
  weaker type
  \[
    \idt{h} :
    \tyARR
    {
      \tyS{\tyB}
    }
    {
      \tyW{\upw{2}{1}}{\tyS{\tyB}},
    }
  \]
  assuming one changes the type annotation for~\idt{tm'}.
\end{example}

\subsection{Effective Time Warps}

The main body of this paper manipulates time warps as abstract mathematical
objects.
In an implementation, one may choose a subset of time warps enjoying finite
representations and equipped with computable operations.
We will say that such a subset is~\emph{effective}.

\begin{defn}[Effectivity]
  A set~$\EFF$ of time warps is~\emph{effective} when it
  contains~$\wID$,~$\wZERO$,~$\wPRED$, and~$\wOM$, is closed under composition,
  division, binary suprema, and binary infima, and is equipped with effective
  procedures for
  \begin{itemize}
  \item
    the aforementioned operations;
  \item
    computing~$p(n)$ for~$p \in \EFF, n \in \NI$;
  \item
    deciding the pointwise order between its elements.
  \end{itemize}
\end{defn}
Any effective set of time warps~$\EFF$ determines a submonoid of~$\WARP$.
Furthermore, the decidability of~$\le$ entails the decidability of equality
between the elements of~$\EFF$ by antisymmetry.
As a consequence, the big-step evaluation judgment
from~\Refsec{operational-semantics} and abstract type-checking algorithm
from~\Refsec{algorithmic-type-checking}, restricted to an effective set of time
warps, become implementable.

\subsection{Ultimately Periodic Sequences}

In~\Refsec{examples}, we have represented certain time warps as running sums of
ultimately periodic sequences.
We now study this classic
idea~\cite{CohenDurantonEisenbeisPagettiPlateau-POPL-2006,Plateau-2010} more
formally, showing that the set~$\PER$ of time warps representable as such
sequences is effective.

Given two finite lists~$u,v$ of elements of~$\NI$, with~$v$ non-empty, we denote
by~$\upw{u}{v}$ the ultimately periodic sequence starting with~$u$ and
continuing with~$v$ repeated~\emph{ad infinitum}.
We say that~$u$ and~$v$ are the~\emph{prefix} and~\emph{periodic pattern}
of~$\upw{u}{v}$, respectively.
Let~$\upw{u}{v}[n]$ denote the~$n$-th element of the sequence, starting at~$0$.
\begin{defn}
  Every ultimately periodic sequence~$\upw{u}{v}$ gives rise to a time warp
  characterized in a unique way by
  \begin{equation}
    \upw{u}{v}(1+n)
    =
    \upw{u}{v}[n] + \upw{u}{v}(n).
  \end{equation}
\end{defn}

Distinct prefix/periodic pattern pairs can represent the same ultimately
periodic sequence; for example~$\pw{1\,0}$,~$\pw{1\,0\,1\,0}$,
and~$\upw{1}{0\,1}$ all represent the same sequence.
Furthermore, distinct ultimately periodic sequence can represent the same time
warp in the presence of~$\omega$; for
example,~$\upw{\omega}{0}$,~$\upw{\omega}{1}$, and~$\pw{\omega}$ all
represent~$\wOM$.
This gives rise to an equivalence relation between prefix/periodic pattern
pairs; we never distinguish between equivalent pairs.

\begin{prop}
  The set of time warps~$\PER$ is effective.
\end{prop}

Rather than give a proof of this statement, we will provide intuitions and
examples.

\paragraph{Common Representations}

The time warps~$\wID$,~$\wZERO$,~$\wPRED$, and~$\wOM$ are respectively
represented by~$\pw{1}$,~$\pw{0}$,~$\upw{0}{1}$, and~$\pw{\omega}$.
The time warp~$\wSUCC$, which corresponds to~$n \mapsto n+1$~($\symSOONER$
in~\cite{BirkedalMogelbergSchwinghammerStovring-2012}), is represented
by~$\upw{2}{1}$.
Any constant time warp~$n \mapsto c$ for~$0 < n < \omega$ is represented
by~$\upw{c}{0}$.

\paragraph{Composition}

To show that~$\PER$ is closed under composition, we build an infinite
sequence~$s$ representing~$\upw{u_1}{v_1} \opON \upw{u_2}{v_2}$ by
traversing~$\upw{u_1}{v_1}$ and~$\upw{u_2}{v_2}$.
There are two cases, depending on the next element of~$\upw{u_1}{v_1}$, which
we call~$n$.
\begin{itemize}
\item
  If~$n < \omega$, the next element of~$s$ is~$\sum_{i = 1}^n m_i$
  with~$m_1,\dotsc,m_n$ the next~$n$ elements of~$\upw{u_2}{v_2}$.
  We then continue building the rest of~$s$ recursively,
  dropping~$n,m_1,\dotsc,m_n$.

\item
  If~$n = \omega$, the next element of~$s$ is the sum of all the remaining
  elements of~$\upw{u_2}{v_2}$.
  The rest of~$s$ is~$\pw{0}$.
\end{itemize}
Why is~$s$ ultimately periodic?
Let us write~$\lLENGTH{l}$ for the length of a list of numbers~$l$
and~$\lWEIGHT{l}$ for the sum of its elements.
Clearly, if there is at least one ocurrence of~$\omega$ in~$\upw{u_1}{v_1}$,
then~$s$ is ultimately periodic.
Otherwise, one can always find~$(u_3,v_3)$ and~$(u_4,v_4)$ such
that~$\upw{u_1}{v_1} =
\upw{u_3}{v_3}$,~$\upw{u_2}{v_2} =
\upw{u_4}{v_4}$,~$\lWEIGHT{u_3} = \lLENGTH{u_4}$
and~$\lWEIGHT{v_3} = \lLENGTH{v_4}$ by unfolding the prefixes and repeating
the periodic patterns of~$\upw{u_1}{v_1}$ and~$\upw{u_2}{v_2}$ as much as
required.
The new words represent the same sequences and thus give rise to the same~$s$,
and it follows from their definition that~$s$ is ultimately periodic with a
prefix of length~$\lLENGTH{u_3}$ and a periodic pattern of
length~$\lLENGTH{v_3}$.

The following examples illustrate composition in~$\PER$.
\begin{mathpar}
  \pw{3} \opON\, \pw{2} = \pw{2} \opON\, \pw{3} = \pw{6}
  \and
  \pw{1\,0} \opON\, \pw{0\,1} = \pw{0\,0\,1\,0}
  \and
  \pw{0\,1} \opON\, \pw{1\,0} = \pw{0\,1\,0\,0}
  \and
  \pw{2} \opON\, \pw{1\,0} = \pw{2} \opON\, \pw{0\,1} = \pw{1}
  \and
  \upw{0}{2} \opON\, \pw{3\,0\,1} = \upw{0}{3\,4\,1}
  \and
  \upw{2}{1} \opON\, \upw{0}{1} = \pw{1}
  \and
  \pw{2\,0} \opON\, \pw{2\,0} = \pw{2\,0}
  \and
  \pw{\omega} \opON\, \pw{1\,0} = \upw{\omega}{0} = \pw{\omega}
  \and
  \pw{0\,1} \opON\, \pw{\omega}  = \upw{0\,\omega}{0} = \upw{0}{\omega}
  \and
  \pw{\omega} \opON\, \pw{0} = \pw{0}
\end{mathpar}

\paragraph{Division}

The result of a time warp division~$\upw{u_1}{v_1} \wD\, \upw{u_2}{v_2}$ is more
complex to build than a composition.
Intuitively, one produces new elements in the resulting sequence according to
the next element~$n$ of~$\upw{u_2}{v_2}$, accumulating the numbers present
in~$\upw{u_1}{v_1}$.
If~$n > 0$, then one outputs the current value of the accumulator,
followed by~$n-1$ zeroes; if~$n = \omega$, the process stops.
If~$n = 0$, then one adds the current element in~$\upw{u_1}{v_1}$ to the
accumulator.
Division by zero is a special case.

The following examples illustrate division in~$\PER$.
\begin{mathpar}
  \pw{1} \wD\, \pw{1} = \pw{1}
  \and
  \pw{2} \wD\, \pw{2} = \pw{2\,0}
  \and
  \pw{1\,0} \wD\, \pw{1\,0} = \pw{1}
  \and
  \pw{1} \wD\, \pw{0\,3} = \pw{2\,0\,0}
  \and
  \pw{4\,0} \wD\, \pw{1\,3} = \pw{4\,0\,0\,0}
  \and
  \pw{0} \wD\, \pw{0} = \pw{\omega}
  \and
  \pw{3} \wD\, \pw{\omega} = \upw{3}{0}
\end{mathpar}

\paragraph{Evaluation}

A naive way to evaluate~$(\upw{u}{v})(n)$ is to compute the sum of the first~$n$
elements of~$\upw{u}{v}$.
%

\paragraph{Ordering}

We have~$\upw{u_1}{v_1} \le \upw{u_2}{v_2}$ if and only
if~$\upw{u_1}{v_1}(n) \le \upw{u_2}{v_2}(n)$ for
all~$1 \le n \le max(|u_1|, |u_2|) + lcm(|v_1|, |v_2|)$.

\paragraph{Infima and Suprema}

Binary infima and suprema can be computed using the above characterization of
the ordering between elements of~$\PER$.

\subsection{Implementation}

We have implemented the elaboration algorithm described in this paper,
restricted to~$\PER$.
Our prototype is available at the address
\begin{center}
  \small
  \url{https://github.com/adrieng/pulsar}.
\end{center}
Example programs, including the ones discussed in~\Refsec{introduction}
and~\Refsec{examples}, and~\Refsec{conclusion}, can be found in the file
\begin{center}
  \small
  \url{https://github.com/adrieng/pulsar/blob/master/examples/streams.pul}.
\end{center}

  \section{Selected Proofs}
\label{app:proofs}

\subsection{The Calculus}

\paragraph{Notations}

We say that~$\Gamma'$ is a~\emph{subcontext} of~$\Gamma$ if~$\Gamma'$
is~$\Gamma$ with zero or more bindings removed~(but not permuted).
Given a context~$\Gamma$ and a finite set of variables~$X$, we
write~$\Restr{\Gamma}{X}$ for the largest subcontext of~$\Gamma$ such
that~$\dom{\Restr{\Gamma}{X}} \subseteq X$.
Given a context~$\Gamma$ and a time warp~$p$, we write~$\minusW{\Gamma}{p}$ for
the largest context~(for the subcontext ordering) such
that~$\tyW{p}{(\minusW{\Gamma}{p})}$ is a subcontext of~$\Gamma$.

We write~$d ::: \jugT{\Gamma}{e}{\tau}$ when~$d$
is a derivation of~$\jugT{\Gamma}{e}{\tau}$.

\paragraph{Type-Checking Explicit Terms}

\begin{figure}
  \begin{mathpar}
    \fbox{$\algE{\Gamma}{e}{\tau}$}

    \IRULE[EAlgVar]
    {
      \Gamma(x) = \tau
    }
    {
      \algE
      {\Gamma}
      {x}
      {\tau}
    }

    \IRULE[EAlgFun]
    {
      \algE
      {\Gamma, x : \tau_1}
      {e}
      {\tau_2}
    }
    {
      \algE
      {\Gamma}
      {\synFUN{x}{\tau_1}{e}}
      {\tyARR{\tau_1}{\tau_2}}
    }

    \IRULE[EAlgApp]
    {
      \algE
      {\Gamma}
      {e_1}
      {\tyARR{\tau_1}{\tau_2}}
      \\
      \algE
      {\Gamma}
      {e_2}
      {\tau_1}
    }
    {
      \algE
      {\Gamma}
      {e_1~e_2}
      {\tau_2}
    }

    \IRULE[EAlgPair]
    {
      (
      \algE
      {\Gamma}
      {e_i}
      {\tau_i}
      )_{i \in \{1,2\}}
    }
    {
      \algE
      {\Gamma}
      {(e_1, e_2)}
      {\tyPROD{\tau_1}{\tau_2}}
    }

    \IRULE[EAlgProj$_{i \in \{1, 2\}}$]
    {
      \algE
      {\Gamma}
      {e}
      {\tyPROD{\tau_1}{\tau_2}}
    }
    {
      \algE
      {\Gamma}
      {\synPROJ{i}{e}}
      {\tau_i}
    }

    \IRULE[EAlgInj$_1$]
    {
      \algE
      {\Gamma}
      {e}
      {\tau_1}
    }
    {
      \algE
      {\Gamma}
      {\synINJ{1}{\tau_2}{e}}
      {\tySUM{\tau_1}{\tau_2}}
    }

    \IRULE[EAlg$_2$]
    {
      \algE
      {\Gamma}
      {e}
      {\tau_2}
    }
    {
      \algE
      {\Gamma}
      {\synINJ{2}{\tau_1}{e}}
      {\tySUM{\tau_1}{\tau_2}}
    }

    \IRULE[EAlgCase]
    {
      \algE
      {\Gamma}
      {e}
      {\tySUM{\tau_1}{\tau_2}}
      \\
      (
      \algE
      {\Gamma, x_i : \tau_i}
      {e_i}
      {\tau}
      )_{i \in \{ 1, 2 \}}
    }
    {
      \algE
      {\Gamma}
      {\synCASE{e}{x_1}{e_1}{x_2}{e_2}}
      {\tau}
    }

    \IRULE[EAlgHead]{
      \algE
      {\Gamma}
      {e}
      {\tyS{\tau}}
    }{
      \algE
      {\Gamma}
      {\synHEAD{e}}
      {\tau}
    }

    \IRULE[EAlgTail]{
      \algE
      {\Gamma}
      {e}
      {\tyS{\tau}}
    }{
      \algE
      {\Gamma}
      {\synTAIL{e}}
      {\tyW{\wPRED}{\tyS{\tau}}}
    }

    \IRULE[EAlgCons]{
      \algE
      {\Gamma}
      {e_1}
      {\tau}
      \\
      \algE
      {\Gamma}
      {e_2}
      {\tyW{\wPRED}{\tyS{\tau}}}
    }{
      \algE
      {\Gamma}
      {\synCONS{e_1}{e_2}}
      {\tyS{\tau}}
    }

    \IRULE[EAlgRec]
    {
      \algE
      {\Gamma, x : \tyW{\wPRED}{\tau}}
      {e}
      {\tau}
    }
    {
      \algE
      {\Gamma}
      {\synREC{x}{\tau}{e}}
      {\tau}
    }

    \IRULE[EAlgWarp]
    {
      \algE
      {\minusW{\Gamma}{p}}
      {e}
      {\tau}
    }
    {
      \algE
      {\Gamma}
      {\synBY{p}{e}}
      {\tyW{p}{\tau}}
    }

    \IRULE[EAlgConst]
    {
      s \in \SCAL_\nu
    }
    {
      \algE
      {\Gamma}
      {s}
      {\nu}
    }
    \\

    \IRULE[EAlgSubR]
    {
      \algE
      {\Gamma}
      {e}
      {\tau}
      \\
      \jugC
      {\alpha}
      {\tau}
      {\tau'}
    }
    {
      \algE
      {\Gamma}
      {\synCOER{e}{\alpha}}
      {\tau'}
    }

    \IRULE[EAlgSubL]
    {
      \jugC
      {\beta}
      {\Restr{\Gamma}{\dom{\beta}}}
      {\Gamma'}
      \\
      \algE
      {\Gamma'}
      {e}
      {\tau}
    }
    {
      \algE
      {\Gamma}
      {\synCOEL{\beta}{e}}
      {\tau}
    }
  \end{mathpar}
  \caption{Type-Checking Explicit Terms}
  \label{fig:type-checking-explicit-terms}
\end{figure}

The algorithmic type-checking judgment~$\algE{\Gamma}{e}{\tau}$ takes a
context~$\Gamma$ and an explicit term~$e$ and returns a type~$\tau$.
Its rules are given in~\Reffig{type-checking-explicit-terms}.

\paragraph{Metatheoretical Results}

\begin{prop}[Determinism of Explicit Type-Checking]
  \label{prop:determinism-of-explicit-type-checking}
  If~$\algE{\Gamma}{e}{\tau}$ and~$\algE{\Gamma}{e}{\tau'}$ then~$\tau = \tau'$.
\end{prop}

\begin{proof}
  Immediate, the judgment is syntax-directed.
\end{proof}

\begin{prop}[Soundness of Explicit Type-Checking]
  \label{prop:soundness-of-explicit-type-checking}
  If~$\algE{\Gamma}{e}{\tau}$ then~$\jugT{\Gamma}{e}{\tau}$.
\end{prop}

\begin{proof}
  By induction on the derivation.
  \begin{itemize}
  \item
    Case~\rn{EAlgVar}:
    since~$x \in \dom{\Gamma}$,~$\Gamma$ must be of the
    form~$\Gamma', x : \tau, x_1 : \tau_1, \dotsc, x_n : \tau_n$ for a
    certain~$\Gamma'$.
    Let us denote~$\sigma_w \in \structMap{\Gamma', x : \tau}{\Gamma}$ the
    inclusion map of~$\dom{\Gamma',x : \tau}$ into~$\dom{\Gamma}$.
    We conclude by deriving~$\jugT{\Gamma', x : \tau}{x}{\tau}$ by~\Refrule{Var}
    and then applying~\Refrule{Struct} with~$\sigma_w$.

  \item
    Cases~\rn{EAlgFun} to~\rn{EAlgRec},~\rn{EAlgConst},~\rn{EAlgSubR}:~immediate
    application of induction hypotheses.

  \item
    Case~\rn{EAlgWarp}:
    by definition~$\tyW{p}{(\minusW{\Gamma}{p})}$ is a subcontext of~$\Gamma$.
    Let us
    denote~$\sigma_w \in \structMap{\tyW{p}{(\minusW{\Gamma}{p})}}{\Gamma}$ the
    corresponding inclusion map.
    By the induction hypothesis and~\Refrule{Warp}, we
    have~$\jugT{\tyW{p}{\minusW{\Gamma}{p}}}{\synBY{e}{p}}{\tyW{p}{\tau}}$.
    We conclude by applying~\Refrule{Struct} with~$\sigma_w$.

  \item
    Case~\rn{EAlgSubL}:
    again,~$\Restr{\Gamma}{\dom{\beta}}$ is a subcontext of~$\Gamma$ and we
    denote by~$\sigma_w \in \structMap{\Restr{\Gamma}{\dom{\beta}}}{\Gamma}$ the
    corresponding inclusion map.
    By the induction hypothesis, we have~$\jugT{\Gamma'}{e}{\tau}$.
    Thus, we
    derive~$\jugT{\Restr{\Gamma}{\dom{\beta}}}{\synCOEL{\beta}{e}}{\tau}$
    by~\Refrule{SubL}.
    We conclude by applying~\Refrule{Struct} with~$\sigma_w$.
  \end{itemize}
\end{proof}

Abusing notation, we write~$d ::: \algE{\Gamma}{e}{\tau}$ for the derivation
built in the above proof.
It is exactly the canonical derivation described
in~\Refprop{canonical-derivations}.
Such a derivation can always be built.

\begin{prop}[Completeness of Explicit Type-Checking]
  \label{prop:completeness-of-explicit-type-checking}
  If~$\jugT{\Gamma}{e}{\tau}$ then~$\algE{\Gamma}{e}{\tau}$.
\end{prop}

\begin{prop}[Uniqueness of Types for Explicit Terms]
  For any fixed~$\Gamma$ and~$e$, there is at most one type~$\tau$ such
  that~$\jugT{\Gamma}{e}{\tau}$ holds.
\end{prop}

\begin{proof}
  Immediate consequence of~\Refprop{completeness-of-explicit-type-checking}
  and~\Refprop{determinism-of-explicit-type-checking}.
\end{proof}

\subsection{Operational Semantics}

\begin{figure}
  \footnotesize
  \begin{mathpar}
    \fbox{$\reaV{\tau}{n} \subseteq \mi{Val}$}
    \and
    \fbox{$\reaC{\Gamma}{n} \subseteq \mi{Env}$}
    \and
    \fbox{$\reaE{\Gamma}{\tau} \subseteq \mi{ETerm}$}
  \end{mathpar}
  \begin{align*}
    \reaV{\tau}{0} & =
    \mi{Val}
    \\
    \reaV{\nu}{n+1}
    & =
    \SCAL_\nu
    \\
    \reaV{\tyS{\tau}}{n+1}
    & =
    \{ \vC{v_1}{v_2} \mid
    v_1 \in \reaV{\tau}{n+1} \wedge v_2 \in \reaV{\tyS{\tau}}{n}
    \}
    \\
    \reaV{\tyPROD{\tau_1}{\tau_2}}{n+1}
    & =
    \{ (v_1, v_2) \mid
    v_1 \in \reaV{\tau_1}{n+1} \wedge v_2 \in \reaV{\tyS{\tau_2}}{n+1}
    \}
    \\
    \reaV{\tySUM{\tau_1}{\tau_2}}{n+1}
    & =
    \{ \vINJ{i}{v} \mid
    v \in \reaV{\tau_i}{n+1}
    \}
    \\
    \reaV{\tyARR{\tau_1}{\tau_2}}{n+1}
    & =
    \{ \vCLO{x}{e}{\gamma} \mid
    \begin{array}[t]{@{}l}
      \forall m \le n+1, \forall v_1 \in \reaV{\tau_1}{m},
      \exists \gamma',
      \jugTR{\gamma}{m}{\gamma'}
      \\
      \wedge \exists v_2,
      \jugEV{e}{\gamma'[x \mapsto v_1]}{m}{v_2}
      \wedge\, v_2 \in \reaV{\tau_2}{m}
      \}
    \end{array}
    \\
    \reaV{\tyW{p}{\tau}}{n+1}
    & =
    \{ \vW{p}{v} \mid
    v \in \reaV{\tau}{p(n+1)}
    \}
    \\
    \reaV{\tau}{\omega}
    & =
    \{ \vTHUNK{e}{\gamma} \mid
    \begin{array}[t]{@{}l}
      \forall m < \omega,
      \exists v,
      \\
      \jugTR{\vTHUNK{e}{\gamma}}{m}{v}
      \wedge v \in \reaV{\tau}{m}
      \}
    \end{array}
    \\
    \reaC{\Gamma}{n} & =
    \{ \gamma \mid
    \dom{\gamma} = \dom{\Gamma}
    \wedge
    \forall x \in \dom{\gamma},
    \gamma(x) \in \reaV{\Gamma(x)}{n}
    \}
    \\
    \reaE{\Gamma}{\tau} & =
    \{ e \mid
    \forall n \in \NI,
    \forall \gamma \in \reaC{\Gamma}{n},
    \exists v,
    \jugEV{e}{\gamma}{n}{v}
    \wedge
    v \in \reaV{\tau}{n}
    \}
  \end{align*}
  \caption{Realizability Predicates for Totality}
  \label{fig:operational-realizability-predicates}
\end{figure}

The proof of~\Refthm{totality} relies on three realizability predicates defined
in~\Reffig{operational-realizability-predicates}.
They follow the usual structure of step-indexed logical relation,
defining~$\reaV{\tau}{n}$ by well-founded induction over the lexicographic
order~$\mi{Lex}(<_t,<_n) \subseteq (\mi{Types} \times (\NI))^2$, where~$<_t$ is
the proper subterm ordering and~$<_n$ the canonical ordering on~$\NI$.
Most clauses are standard.
We prove the fundamental property of logical relations,
obtaining~\Refthm{totality} as an immediate corollary.

\begin{lemma}[Fundamental Property]
  If~$\jugT{\Gamma}{e}{\tau}$ then~$e \in \reaE{\Gamma}{\tau}$.
\end{lemma}

As before, this proof relies on additional lemmas for each auxiliary judgment.
Their statements are similar to the ones used for type safety, replacing
syntactic typing with realizability.

\subsection{Denotational Semantics}

\paragraph{Interpretation of Structure Maps}

Consider an arbitrary structure map~$\sigma \in \structMap{\Gamma}{\Gamma'}$.
Let us write~$\Gamma$ as~$x_1 : \tau_1, \dotsc, x_n : \tau_n$ and~$\Gamma'$
as~$x_1' : \tau_1', \dotsc, x_m' : \tau_m'$.
Since contexts are interpreted as tuples, we
define~$\sem{\sigma} : \sem{\Gamma'} \to \sem{\Gamma}$
as~$\langle f_1, \dotsc, f_n \rangle$
with~$f_i : \sem{\Gamma'} \to \sem{\tau_i}$ the~$j$th projection out
of~$\sem{\Gamma}$ for~$j$ such that~$\sigma(x_i) = x_j'$.

\paragraph{Coherence for Explicit Terms}

\begin{prop}
  \label{prop:explicit-type-checking-struct}
  If~$\algE{\Gamma}{e}{\tau}$ and~$\sigma \in \structMap{\Gamma}{\Gamma'}$
  then~$\algE{\Gamma'}{\sigma[e]}{\tau}$ with
  \[
    \sem{d_1 ::: \algE{\Gamma'}{\sigma[e]}{\tau}}
    =
    \sem{d_2 ::: \algE{\Gamma}{e}{\tau}}
    \circ
    \sem{\sigma}.
  \]
\end{prop}

\begin{prop}
  \label{prop:explicit-type-checking-coherence}
  For any~$d_1 ::: \jugT{\Gamma}{e}{\tau}$ we
  have~$d_2 ::: \algE{\Gamma}{e}{\tau}$ with
  \[
    \sem{d_1 ::: \jugT{\Gamma}{e}{\tau}}
    =
    \sem{d_2 ::: \algE{\Gamma}{e}{\tau}}.
  \]
\end{prop}

\begin{proof}
  By induction on~$d_1$.
  \begin{itemize}
  \item
    Case~\rn{Var}:
    since~$(\Gamma, x : \tau) = \tau$,~$\algE{\Gamma, x : \tau}{x}{\tau}$ holds
    by definition.
    The corresponding derivation~$d_2$ is formed of~\Refrule{Var}
    and~\Refrule{Struct} with~$\sigma_w = \id$.
    Hence~$\sem{d_2} = \sem{d_1} \circ \id = \sem{d_1}$.

  \item
    Cases~\rn{Fun} to~\rn{Rec},~\rn{Const},~\rn{SubR}:~immediate application of
    induction hypotheses.

  \item
    Case~\rn{Warp}:~we have~$d_1' ::: \jugT{\Gamma}{e}{\tau}$.
    The induction hypothesis gives~$d_2' ::: \algE{\Gamma}{e}{\tau}$
    with~$\sem{d_1'} = \sem{d_2'}$.
    Since~$\minusW{(\tyW{p}{\Gamma})}{p} = \Gamma$, we
    have~$\algE{\tyW{p}{\Gamma}}{\synBY{e}{p}}{\tau}$ with a canonical
    derivation~$d_2$ ending with an instance of~\Refrule{Struct}
    with~$\sigma_w = \id$.
    We conclude as in the~\rn{Var} case.

  \item
    Case~\rn{SubL}:
    we have~$d_1' = \jugT{\Gamma'}{e}{\tau}$
    and~$\jugC{\beta}{\Gamma}{\Gamma'}$.
    The induction hypothesis gives~$d_2' ::: \algE{\Gamma'}{e}{\tau}$.
    Since~$\jugC{\beta}{\Gamma}{\Gamma'}$, we
    have~$\dom{\Gamma} = \dom{\Gamma'} \subseteq \dom{\beta}$ by definition.
    Hence~$\Restr{\Gamma}{\dom{\beta}} = \Gamma$ and we
    have~$\algE{\Gamma}{\synCOEL{\beta}{e}}{\tau}$ with a canonical
    derivation~$d_2$ using~$d_2'$ and ending with an instance
    of~\Refrule{Struct} with~$\sigma_w = \id$.
    We conclude as in the~\rn{Var} case.

  \item
    Case~\rn{Struct}:~we have~$d_1' ::: \jugT{\Gamma}{e}{\tau}$.
    The induction hypothesis gives~$d_2' ::: \algE{\Gamma}{e}{\tau}$
    with~$\sem{d_1'} = \sem{d_2'}$.
    Applying~\Refprop{explicit-type-checking-struct}, we
    obtain~$d_2 = \algE{\Gamma'}{\sigma[e]}{\tau}$
    with~$\sem{d_2'} = \sem{d_2} \circ \sem{\sigma}$.
    We conclude by transitivity.

  \end{itemize}
\end{proof}

\Refprop{explicit-type-checking-coherence} directly
implies~\Refprop{completeness-of-explicit-type-checking}.

\paragraph{Contexts}

We define contexts~$C$ in the usual way, as explicit terms with a single hole,
with a plugging operation~$C[-]$.
We write
\[
  C : (\Gamma \vdash \tau) \to (\Gamma' \vdash \tau')
\]
when the context~$C$ maps terms~$\jugT{\Gamma}{e}{\tau}$ to a
terms~$\jugT{\Gamma'}{C[e]}{\tau'}$.
Strictly speaking, we should introduce a typing judgment for contexts, but its
definition is straightforward and we thus elide it.

Every context~$C : (\Gamma \vdash \tau) \to (\Gamma' \vdash \tau')$ defines a
morphism~$\sem{C} : \sem{\tau}^{\sem{\Gamma}} \to \sem{\tau'}^{\sem{\Gamma'}}$
of~\TT{}.
The denotational semantics is compositional as expected.
\begin{prop}
  \label{prop:compositionality}
  If~$C : (\Gamma \vdash \tau) \to (\Gamma' \vdash \tau')$
  and~$\jugT{\Gamma}{e}{\tau}$, then
  \[
    \sem{C} \circ \semCURRY{\sem{\jugT{\Gamma}{e}{\tau}}}
    =
    \sem{\jugT{\Gamma'}{C[e]}{\tau'}}.
  \]
\end{prop}

\paragraph{Operational Equivalence}

Programs can be discriminated by observing ground values at the first step,
assuming the set of ground values contains at least booleans.
Given two explicit terms~$\jugT{\Gamma}{e_1}{\tau}$
and~$\jugT{\Gamma}{e_2}{\tau}$, we define operational equivalence as follows.
\begin{align*}
  &
  \CTXEQ{\Gamma}{e_1}{e_2}{\tau}
  \\
  \defeq &
  \forall C : (\Gamma \vdash \tau) \to (\ctxempty \vdash \nu),
  \jugEV{C[e_1]}{\emptyset}{1}{s}
  \Leftrightarrow
  \jugEV{C[e_2]}{\emptyset}{1}{s}
\end{align*}




\paragraph{Adequacy}

In order to relate our operational and denotational semantics, we embed the
operational semantics inside~$\TT$.
For every type~$\tau$, we define a presheaf~$\opsem{\tau}$ by
\begin{align*}
  \opsem{\tau}(n) & \defeq \{ v \in \mi{Val} \mid \jugV{v}{\tau}{n} \}
  \\
  \opsem{\tau}(n \le m) & \defeq \vTRUNC{m}{-}
\end{align*}
where~$\vTRUNC{m}{v} \defeq \{ v' \mid \jugTR{v}{m}{v'} \}$.
The results from~\Refsec{operational-semantics} imply that~$\opsem{\tau}$ is
well-defined:~determinism, type safety, and totality imply that~$\vTRUNC{m}{-}$
is a function between the proper sets, and~\Refprop{functoriality-of-truncation}
indeed corresponds to functoriality.
We apply the same idea to typing contexts, building a presheaf~$\opsem{\Gamma}$
of environments.
For every well-typed explicit term~$\jugT{\Gamma}{e}{\tau}$, we build a natural
transformation
\[
\begin{array}{l@{~}c@{~}l}
  \opsem{e} & : & \opsem{\Gamma} \to \opsem{\tau}
  \\
  \opsem{e}_n & \defeq & \gamma \mapsto \{ v \mid \jugEV{e}{\gamma}{n}{v} \}.
\end{array}
\]
Again, the results from~\Refsec{operational-semantics} imply that this is
well-defined.
In particular,~\Refprop{monotonicity} corresponds to naturality.

We then introduce a logical relation~$\IMPL{\tau}$ which, intuitively, defines
what it means for a point of~$\opsem{\tau}$ to implement a point
of~$\sem{\tau}$.
Its definition is straightforward, given the close correspondance between values
and denotational elements, and thus we elide it.
We lift it to contexts/value pairs, obtaining a relation~$\IMPLM{\Gamma}{\tau}$
between~$\opsem{\Gamma} \to \opsem{\tau}$ and~$\sem{\Gamma} \to \sem{\tau}$, and
prove the fundamental property.
\begin{lemma}
  \label{lemm:fundamental-lemma-2}
  If~$\jugT{\Gamma}{e}{\tau}$ then~$\opsem{e} \IMPLM{\Gamma}{\tau} \sem{e}$.
\end{lemma}
Now, given two well-typed terms~$\jugT{\Gamma}{e_1,e_2}{\tau}$, let us
write~$\EQIMPL{\Gamma}{e_1}{e_2}{\tau}$ when~$e_1$ and~$e_2$ implement the same
denotation, that is,~if there exists~$f : \sem{\Gamma} \to \sem{\tau}$ such that
both~$\opsem{e_1} \IMPLM{\Gamma}{\tau} f$
and~$\opsem{e_2} \IMPLM{\Gamma}{\tau} f$.
Such terms are indistinguishable in~\Core{}.
This relation is a congruence.
\begin{lemma}
  \label{lemm:congruence}
  If~$\EQIMPL{\Gamma}{e}{e'}{\tau}$
  and~$C : (\Gamma \vdash \tau) \to (\Gamma' \vdash \tau')$,
  then
  \[
    \EQIMPL{\Gamma'}{C[e]}{C[e']}{\tau'}
    .
  \]
\end{lemma}

We obtain adequacy as a corollary of~\Reflemm{fundamental-lemma-2}
and~\Reflemm{congruence}.
If~$\sem{\jugC{\Gamma}{e}{\tau}} = \sem{\jugC{\Gamma}{e}{\tau}}$,
then~$\EQIMPL{\Gamma}{e_1}{e_2}{\tau}$ by~\Reflemm{fundamental-lemma-2}.
We immediately conclude~$\CTXEQ{\Gamma}{e_1}{e_2}{\tau}$ by~\Reflemm{congruence}
and the definition of~$\IMPL{\nu}$.

\subsection{Algorithmic Type Checking}

\subsubsection{Notations}

As in the main body of the paper, we write~$\tau_1 \subty \tau_2$ to indicate
the mere existence of some coercion~$\jugC{\alpha}{\tau_1}{\tau_2}$, and
similarly for~$\tau_1 \equiv \tau_2$.
We also write~$\jugEC{\alpha_1}{\alpha_2}{\tau}{\tau'}$ for
\[
  \sem{\jugC{\alpha_1}{\tau}{\tau'}} = \sem{\jugC{\alpha_2}{\tau}{\tau'}}
\]
and similarly~$\jugE{\Gamma}{e_1}{e_2}{\tau}$ for
\[
  \sem{\jugT{\Gamma}{e_1}{\tau}}
  =
  \sem{\jugT{\Gamma}{e_2}{\tau}}
  .
\]

\subsubsection{Properties of Type Normalization}

\begin{prop}[Idempotence of Normalization]
  We have
  \begin{align}
    \NORM{\NORM{\tau}} & = \NORM{\tau}.
  \end{align}
\end{prop}

\begin{proof}
  By induction on the size of~$\tau$.
\end{proof}

\begin{prop}[Confluence of Normalization]
  \label{prop:confluence}
  We have
  \begin{align}
    \NORM{\tyARR{\tau_1}{\tau_2}}
    & =
    \NORM{\tyARR{\NORM{\tau_1}}{\tau_2}}
    =
    \NORM{\tyARR{\tau_1}{\NORM{\tau_2}}},
    \\
    \NORM{\tyPROD{\tau_1}{\tau_2}}
    & =
    \NORM{\tyPROD{\NORM{\tau_1}}{\tau_2}}
    =
    \NORM{\tyPROD{\tau_1}{\NORM{\tau_2}}},
    \\
    \NORM{\tySUM{\tau_1}{\tau_2}}
    & =
    \NORM{\tySUM{\NORM{\tau_1}}{\tau_2}}
    =
    \NORM{\tySUM{\tau_1}{\NORM{\tau_2}}},
    \\
    \NORM{\tyW{p}{\tau}}
    & =
    \NORM{\tyW{p}{\NORM{\tau}}}.
  \end{align}
\end{prop}

\begin{prop}
  \label{prop:norm-id}
  For any~$\tau$, we
  have~$\jugEC{\NORMin{\tau};\NORMout{\tau}}{\coeID}{\tau}{\tau}$
  and~$\jugEC{\NORMout{\tau};\NORMin{\tau}}{\coeID}{\NORM{\tau}}{\NORM{\tau}}$.
\end{prop}

\begin{proof}
  By induction on the size of~$\tau$.
\end{proof}

\begin{lemma}
  \label{lemm:norm}
  For any~$\jugC{\alpha_1}{\tau_1}{\tau_2}$, if there
exists~$\jugC{\alpha_2}{\tau_2}{\tau_1}$, then~$\NORM{\tau_1} = \NORM{\tau_2}$.
Moreover~$\jugEC{\alpha_1}{\NORMin{\tau_1};\NORMout{\tau_2}}{\tau_1}{\tau_2}$
and~$\jugEC{\alpha_2}{\NORMin{\tau_2};\NORMout{\tau_1}}{\tau_2}{\tau_1}$.
\end{lemma}

\begin{proof}
  By induction on~$\alpha_1$.
\end{proof}

\begin{coro}
  \label{coro:norm}
  For any~$\jugCE{\alpha_1}{\alpha_2}{\tau_1}{\tau_2}$, we
  have~$\jugEC{\alpha_1;\alpha_2}{\coeID}{\tau_1}{\tau_1}$
  and~$\jugEC{\alpha_2;\alpha_1}{\coeID}{\tau_2}{\tau_2}$.
\end{coro}

\begin{proof}
  Immediate consequence of~\Refprop{norm-id} and~\Reflemm{norm}.
\end{proof}

\begin{theorem}
  \label{thm:norm-sound-complete}
  Type normalization is sound and complete for
  equivalence:~$\tau_1 \equiv \tau_2$ iff~$\NORM{\tau_1} = \NORM{\tau_2}$.
\end{theorem}

\begin{proof}
  The right-to-left direction follows
  from~$\jugC{\NORMin{\tau_1};\NORMout{\tau_2}}{\tau_1}{\tau_2}$
  and~$\jugC{\NORMin{\tau_2};\NORMout{\tau_1}}{\tau_2}{\tau_1}$.
  The left-to-right direction follows immediately from~\Reflemm{norm}.
\end{proof}

\subsubsection{Properties of Type Precedence}

\begin{lemma}[Precedence is Reflexive]
  \label{lemm:reflexivity-of-precedence}
  If~$\tau$ is normal,~$\Prec{\tau}{\tau}$ is defined and
  moreover
  \[
    \jugEC{\Prec{\tau}{\tau}}{\coeID}{\tau}{\tau}.
  \]
\end{lemma}

\begin{lemma}[Precedence is Transitive]
  \label{lemm:transitivity-of-precedence}
  If both~$\Prec{\tau_1}{\tau_2}$ and $\Prec{\tau_2}{\tau_3}$ are defined, so
  is~$\Prec{\tau_1}{\tau_3}$, and
  moreover
  \[
    \jugEC
    {\Prec{\tau_1}{\tau_2};\Prec{\tau_2}{\tau_3}}
    {\Prec{\tau_1}{\tau_3}}
    {\tau_1}
    {\tau_3}.
  \]
\end{lemma}

\begin{lemma}[Precedence and Normalization]
  \label{lemm:precedence-normalization}
  If~$\Prec{\tau_1}{\tau_2}$ is defined, so
  is~$\Prec{\NORM{\tau_1}}{\NORM{\tau_2}}$, and moreover
  \[
    \jugEC
    {\Prec{\tau_1}{\tau_2}}
    {
      \NORMin{\tau_1};
      \Prec{\NORM{\tau_1}}{\NORM{\tau_2}};
      \NORMout{\tau_2}
    }
    {\tau_1}
    {\tau_2}.
  \]
\end{lemma}

\subsection{Properties of Algorithmic Subtyping}

\paragraph{Functoriality Properties}

\begin{lemma}[Algorithmic Subtyping is Reflexive]
  \label{lemm:reflexivity-of-algorithmic-subtyping}
  For any type~$\tau$, $\Coe{\tau}{\tau}$ is defined and moreover
  \[
    \jugEC
    {\Coe{\tau}{\tau}}
    {\coeID}
    {\tau}
    {\tau}.
  \]
\end{lemma}

\begin{lemma}[Algorithmic Subtyping is Transitive]
  \label{lemm:transitivity-of-algorithmic-subtyping}
  If~$\Coe{\tau_1}{\tau_2}$ and $\Coe{\tau_2}{\tau_3}$ are defined, so
  is~$\Coe{\tau_1}{\tau_3}$, and
  moreover
  \[
    \jugEC
    {\Coe{\tau_1}{\tau_2};\Coe{\tau_2}{\tau_3}}
    {\Coe{\tau_1}{\tau_3}}
    {\tau_1}
    {\tau_3}.
  \]
\end{lemma}

\begin{lemma}[Algorithmic Subtyping is a Congruence for Arrow Types]
  \label{lemm:algorithmic-subtyping-is-a-congruence-for-arrow-types}
  If both~$\Coe{\tau_3}{\tau_1}$ and~$\Coe{\tau_2}{\tau_4}$ are defined, then so
  is $\Coe{\tyARR{\tau_1}{\tau_2}}{\tyARR{\tau_3}{\tau_4}}$, and moreover
  \[
    \begin{array}{c@{~}l}
      &
      \Coe{\tyARR{\tau_1}{\tau_2}}{\tyARR{\tau_3}{\tau_4}}
      \cong
      \coeARR
      {
        \Coe{\tau_3}{\tau_1}
      }
      {
        \Coe{\tau_2}{\tau_4}
      }
      \\
      : &
      \tyARR{\tau_1}{\tau_2}
      \subty
      \tyARR{\tau_3}{\tau_4}
      .
    \end{array}
  \]
\end{lemma}

\begin{lemma}[Algorithmic Subtyping is a Congruence for Product Types]
  \label{lemm:algorithmic-subtyping-is-a-congruence-for-product-types}
  If both~$\Coe{\tau_1}{\tau_3}$ and~$\Coe{\tau_2}{\tau_4}$ are defined, then so
  is $\Coe{\tyPROD{\tau_1}{\tau_2}}{\tyPROD{\tau_3}{\tau_4}}$, and moreover
  \[
    \begin{array}{c@{~}l}
      &
      \Coe{\tyPROD{\tau_1}{\tau_2}}{\tyPROD{\tau_3}{\tau_4}}
      \cong
      \coePROD
      {
        \Coe{\tau_3}{\tau_1}
      }
      {
        \Coe{\tau_2}{\tau_4}
      }
      \\
      : &
      \tyPROD{\tau_1}{\tau_2}
      \subty
      \tyPROD{\tau_3}{\tau_4}
      .
    \end{array}
  \]
\end{lemma}

\begin{lemma}[Algorithmic Subtyping is a Congruence for Sum Types]
  \label{lemm:algorithmic-subtyping-is-a-congruence-for-sum-types}
  If both~$\Coe{\tau_1}{\tau_3}$ and~$\Coe{\tau_2}{\tau_4}$ are defined, then so
  is $\Coe{\tySUM{\tau_1}{\tau_2}}{\tySUM{\tau_3}{\tau_4}}$, and moreover
  \[
    \begin{array}{c@{~}l}
      &
      \Coe{\tySUM{\tau_1}{\tau_2}}{\tySUM{\tau_3}{\tau_4}}
      \cong
      \coeSUM
      {
        \Coe{\tau_3}{\tau_1}
      }
      {
        \Coe{\tau_2}{\tau_4}
      }
      \\
      : &
      \tySUM{\tau_1}{\tau_2}
      \subty
      \tySUM{\tau_3}{\tau_4}
      .
    \end{array}
  \]
\end{lemma}

\begin{lemma}[Algorithmic Subtyping is a Congruence for Warped Types]
  \label{lemm:algorithmic-subtyping-is-a-congruence-for-warped-types}
  If~$\Coe{\tau_1}{\tau_2}$ is defined, then so is
  $\Coe{\tyW{p}{\tau_1}}{\tyW{p}{\tau_2}}$, and moreover
  \[
    \jugEC
    {\Coe{\tyW{p}{\tau_1}}{\tyW{p}{\tau_2}}}
    {\coeW{p}{\Coe{\tau_1}{\tau_2}}}
    {\tyW{p}{\tau_1}}
    {\tyW{p}{\tau_2}}
    .
  \]
\end{lemma}

\subsubsection{Completeness}

\begin{lemma}[Completness of Algorithmic Subtyping for Invertible Coercions]
  \label{lemm:completeness-of-algorithmic-subtyping-for-invertible-coercions}
  If~$\jugCE{\alpha_1}{\alpha_2}{\tau_1}{\tau_2}$, then~$\Coe{\tau_1}{\tau_2}$
  is defined, and moreover
  \[
    \jugEC
    {\Coe{\tau_1}{\tau_2}}
    {\alpha_1}
    {\tau_1}
    {\tau_2}.
  \]
\end{lemma}

\begin{proof}
  Immediate consequence of~\Reflemm{norm}
  and~\Reflemm{reflexivity-of-precedence}.
\end{proof}

\begin{lemma}[Completness of Algorithmic Subtyping for Delays]
  \label{lemm:completeness-of-algorithmic-subtyping-for-delays}
  For any~$p,q$ and~$\tau$ such
  that~$p \ge q$,~$\Coe{\tyW{p}{\tau}}{\tyW{q}{\tau}}$ is defined, and moreover
  \[
    \jugEC
    {\Coe{\tyW{p}{\tau}}{\tyW{q}{\tau}}}
    {\coeDEL{p}{q}}
    {\tyW{p}{\tau}}
    {\tyW{q}{\tau}}.
  \]
\end{lemma}

\begin{proof}
  By~\Reflemm{reflexivity-of-precedence},~$\Prec{\tau}{\tau}$ is defined and
  equivalent to~$\coeID$.
  Thus, since~$p \ge q$, by definition we have
  \[
    \jugEC
    {\Prec{\tyW{p}{\tau}}{\tyW{q}{\tau}}}
    {(\coeW{p}{\coeID});\coeDEL{p}{q}}
    {\tyW{p}{\tau}}
    {\tyW{q}{\tau}}.
  \]
  By~\Reflemm{precedence-normalization},~$\Prec{\NORM{\tyW{p}{\tau}}}{\NORM{\tyW{q}{\tau}}}$
  is defined, and thus by definition so is~$\Coe{\tyW{p}{\tau}}{\tyW{q}{\tau}}$.
  We have
  \begin{align*}
    \Coe{\tyW{p}{\tau}}{\tyW{q}{\tau}}
    & \cong
    \NORMin{\tyW{p}{\tau}};
    \Prec{\NORM{\tyW{p}{\tau}}}{\NORM{\tyW{q}{\tau}}};
    \NORMout{\tyW{p}{\tau}}
    \\
    & \cong
    \Prec{\tyW{p}{\tau}}{\tyW{q}{\tau}}
    \\
    & \cong
    (\coeW{p}{\coeID});\coeDEL{p}{q}
    \\
    & \cong
    \coeDEL{p}{q}
  \end{align*}
  where the second equation follows from~\Reflemm{precedence-normalization}.
\end{proof}

\begin{theorem}[Completeness of Algorithmic Subtyping]
  \label{thm:completeness-of-algorithmic-subtyping}
  If~$\jugC{\alpha}{\tau_1}{\tau_2}$, then~$\Coe{\tau_1}{\tau_2}$ is defined,
  and moreover
  \[
    \jugEC
    {\Coe{\tau_1}{\tau_2}}
    {\alpha}
    {\tau_1}
    {\tau_2}.
  \]
\end{theorem}

\begin{proof}
  Immediate induction on derivations
  using~Lemmas~\Reflemmshort{reflexivity-of-algorithmic-subtyping},
  \Reflemmshort{transitivity-of-algorithmic-subtyping},
  \Reflemmshort{algorithmic-subtyping-is-a-congruence-for-arrow-types},
  \Reflemmshort{algorithmic-subtyping-is-a-congruence-for-product-types},
  \Reflemmshort{algorithmic-subtyping-is-a-congruence-for-warped-types},
  \Reflemmshort{completeness-of-algorithmic-subtyping-for-invertible-coercions},
  and \Reflemmshort{completeness-of-algorithmic-subtyping-for-delays}.
\end{proof}

\begin{corollary}[Coherence for Coercions]
  \label{coro:coherence-for-coercions}
  For any pair of coercions~$\jugC{\alpha_1}{\tau_1}{\tau_2}$
  and~$\jugC{\alpha_2}{\tau_1}{\tau_2}$, we
  have~$\jugEC{\alpha_1}{\alpha_2}{\tau_1}{\tau_2}$.
\end{corollary}

\begin{proof}
  By~\Refthm{completeness-of-algorithmic-subtyping},~$\Coe{\tau_1}{\tau_2}$ is
  defined and equivalent to both~$\alpha_1$ and~$\alpha_2$.
\end{proof}

\subsection{Context Division}
\label{sec:properties-of-context-division}

\begin{lemma}[Completeness of Type Division]
  \label{lemm:completeness-of-type-division}
  For any~$p$ and~$\tau_1$, we have~$\tau_1 \subty \tyW{p}{(\tau_1 \wD p)}$.
  Moreover, if~$\tau_1 \subty \tyW{p}{\tau_2}$,
  then~$\tau_1 \wD p \subty \tau_2$.
\end{lemma}

\begin{lemma}[Completeness of Context Division]
  \label{lemm:completeness-of-context-division}
  For any~$p$ and~$\Gamma_1$, we
  have~$\Gamma_1 \subty \tyW{p}{(\Gamma_1 \wD p)}$.
  Moreover, if~$\Gamma_1 \subty \tyW{p}{\Gamma_2}$,
  then~$\Gamma_1 \wD p \subty \Gamma_2$.
\end{lemma}

\subsection{Properties of Type Bounds}
\label{sec:properties-of-type-bounds}

\begin{prop}
  $\tau_1 \subty \tau$
  and
  $\tau_2 \subty \tau$
  iff
  $\SUP{\tau_1}{\tau_2} \subty \tau$.
\end{prop}

\begin{prop}
  $\tau \subty \tau_1$
  and
  $\tau \subty \tau_2$
  iff
  $\tau \subty \INF{\tau_1}{\tau_2}$.
\end{prop}

\subsection{Properties of Algorithmic Type Checking}

\begin{lemma}[Elimination of Context Subtyping]
  \label{lemm:elimination-of-context-subtyping}
  If $\Elab{\Gamma'}{t} = (\tau, e)$, then for any
  $\jugC{\beta}{\Gamma}{\Gamma'}$ there exists $\tau^m,e^m,\alpha^m$ such that
  $\Elab{\Gamma}{t} = (\tau^m, e^m)$, $\jugC{\alpha^m}{\tau^m}{\tau}$ and
  \[
    \jugE{\Gamma}{\beta;e}{e^m;\alpha^m}{\tau}.
  \]
\end{lemma}

\begin{proof}
  By induction on~$t$.
  The proof is mostly straightforward manipulations of induction hypotheses, so
  we only detail the crucial case of~\Refrule{AlgWarp}, where
  \begin{mathpar}
    t = \synBY{t'}{p}
    \and
    e = \Coe{\Gamma'}{\tyW{p}{(\Gamma' \wD p)}};(\synBY{e'}{p})
    \and
    \tau = \tyW{p}{\tau'}
  \end{mathpar}
  with~$(\tau', e') = \Elab{\Gamma' \wD p}{t'}$.
  By~\Reflemm{completeness-of-context-division}
  and~\Refthm{completeness-of-algorithmic-subtyping}, we know
  that~$\Coe{\Gamma}{\tyW{p}{(\Gamma \wD p)}}$ is defined and
  that~$\Gamma \wD p \subty \Gamma' \wD p$
  since~$\Gamma \subty \tyW{p}{(\Gamma' \wD p)}$ by transitivity.
  By the induction hypothesis, we obtain~${\tau^m}',{e^m}',{\alpha^m}'$ such
  that~$\Elab{\Gamma \wD p}{t'} = ({\tau^m}',{e^m}')$ and
  \[
    \jugE
    {\Gamma \wD p}
    {\Coe{\Gamma \wD p}{\Gamma' \wD p};e'}
    {{e^m}';{\alpha^m}'}
    {\tau}.
  \]
  Thus, by definition of~$\Elab{\Gamma}{\synBY{t}{p}}$, we must
  have~$\tau^m \defeq \tyW{p}{{\tau^m}'}$
  and~$e^m \defeq \Coe{\Gamma}{\tyW{p}{(\Gamma \wD p)}};(\synBY{{e^m}'}{p})$.
  We take~$\alpha^m\defeq \coeW{p}{{\alpha^m}'}$.
  We have
  \[
  \begin{array}{c@{~}l}
    &
    \beta;
    \Coe{\Gamma'}{\tyW{p}{(\Gamma' \wD p)}};
    (\synBY{e'}{p})
    \\
    \cong
    &
    \Coe{\Gamma}{\tyW{p}{(\Gamma \wD p)}};
    \coeW{p}{\Coe{\Gamma \wD p}{\Gamma' \wD p}};
    (\synBY{e'}{p})
    \\
    \cong
    &
    \Coe{\Gamma}{\tyW{p}{(\Gamma \wD p)}};
    (
    \synBY
    {
      (
      \Coe{\Gamma \wD p}{\Gamma' \wD p};
      e'
      )
    }
    {p}
    )
    \\
    \cong
    &
    \Coe{\Gamma}{\tyW{p}{(\Gamma \wD p)}};
    (
    \synBY
    {
      (
      {e^m}';
      {\alpha^m}'
      )
    }
    {p}
    )
    \\
    \cong
    &
    \Coe{\Gamma}{\tyW{p}{(\Gamma \wD p)}};
    (
    \synBY{{e^m}'}{p}
    );
    \coeW{p}{{\alpha^m}'}
    \\
    \cong
    &
    e^m;
    \alpha^m
  \end{array}
  \]
  where the first equation follows from~\Refcoro{coherence-for-coercions}.
\end{proof}

\begin{lemma}[Algorithmic Type Checking Commutes with Structure Maps]
  \label{lemm:algorithmic-type-checking-commutes-with-structure-maps}
  If~$\Elab{\Gamma}{t} = (\tau, e)$ is defined, then for
  any~$\sigma \in \structMap{\Gamma}{\Gamma'}$ we
  have~$\Elab{\Gamma'}{\sigma[t]} = (\tau, \sigma[e])$.
\end{lemma}

\begin{theorem}[Completeness of Algorithmic Type Checking]
  \label{thm:completeness-of-algorithmic-typing}
  If~$\jugT{\Gamma}{e}{\tau}$, there exists~$e^m,\tau^m,\alpha^m$ such
  that~$\Elab{\Gamma}{\Er{e}} = (\tau^m, e^m)$
  and~$\jugC{\alpha^m}{\tau^m}{\tau}$,
  and moreover
  \[
    \jugE{\Gamma}{e}{e^m;\alpha^m}{\tau}
    .
  \]
\end{theorem}

\begin{proof}
  By induction on the typing derivation.
  The proof relies on the properties and lemmas described
  in~\Refsec{properties-of-context-division},~\Refsec{properties-of-type-bounds},
  and~\Refsec{properties-of-type-bounds}, as well
  as~\Reflemm{elimination-of-context-subtyping}
  and~\Reflemm{algorithmic-type-checking-commutes-with-structure-maps}.
\end{proof}

\begin{corollary}[Coherence for Implicit Terms]
  \label{coro:coherence-for-implicit-terms}
  If~$\jugT{\Gamma}{e}{\tau}$,~$\jugT{\Gamma}{e'}{\tau}$, and~$\Er{e} = \Er{e'}$, then~$\jugE{\Gamma}{e}{e'}{\tau}$.
\end{corollary}

\begin{proof}
  By~\Refthm{completeness-of-algorithmic-typing} we
  have~$e^m,\tau^m,\alpha_1^m,e^m_2,\tau^m_2,\alpha_2^m$ such that
  \begin{mathpar}
    \Elab{\Gamma}{\Er{e}} = (\tau^m_1, e^m_1)
    \and
    \jugC{\alpha^m_1}{\tau^m_1}{\tau}
    \and
    \sem{\jugT{\Gamma}{e_1}{\tau_1}}
    =
    \sem{\jugC{\alpha^m_1}{\tau^m_1}{\tau}}
    \circ
    \sem{\jugT{\Gamma}{e^m_1}{\tau^m_1}}
    \and
    \Elab{\Gamma}{\Er{e}} = (\tau^m_2, e^m_2)
    \and
    \jugC{\alpha^m_2}{\tau^m_2}{\tau}
    \and
    \sem{\jugT{\Gamma}{e_2}{\tau_2}}
    =
    \sem{\jugC{\alpha^m_2}{\tau^m_2}{\tau}}
    \circ
    \sem{\jugT{\Gamma}{e^m_2}{\tau^m_2}}.
  \end{mathpar}
  Since~$\Elab{-}{-}$ is a partial function,~$e^m_1 = e^m_2$
  and~$\tau^m_1 = \tau^m_2$.
  By~\Refcoro{coherence-for-coercions},~$\jugEC{\alpha_1^m}{\alpha_2^m}{\tau^m_1}{\tau}$.
  We conclude by transitivity.
\end{proof}

  \tableofcontents
  \listoffigures
}{}

\end{document}